\documentclass[10pt]{article}

\usepackage{amsmath}
\usepackage{algorithmic}
\usepackage{algorithm}

\usepackage[dvips]{epsfig}
\usepackage{latexsym}
\usepackage{amssymb}
\usepackage{amsmath}
\usepackage{amsthm}
\usepackage{graphics}
\usepackage{subfigure}
\usepackage{hyperref}
\usepackage[toc,page]{appendix}
\usepackage{listings}

\newtheorem{thm}{Theorem}

\newtheorem{lemma}{Lemma}

\begin{document}

\title{\bf Optimal Monitoring and Mitigation of Systemic Risk in Financial
Networks\footnote{We would like to thank Ben Craig, David Marshall,
Celso Brunetti, and seminar participants at the Federal Reserve (Washington,
DC), Federal Reserve Bank of Chicago, Federal Reserve Bank of Cleveland,
and Northwestern University for helpful comments and suggestions.}}
\author{Zhang Li, Xiaojun Lin, Borja Peleato-Inarrea, and
Ilya Pollak\footnote{The authors are with the School of Electrical and Computer Engineering,
Purdue University. Corresponding author's email: ipollak@ecn.purdue.edu.}}
\date{\today} 
\maketitle

\begin{abstract}
This paper studies the problem of optimally allocating a cash injection into
a financial system in distress. Given a one-period borrower-lender network
in which all debts are due at the same time and have the same seniority,
we address the problem of allocating a fixed amount of cash among the nodes
to minimize the weighted sum of unpaid liabilities. Assuming all the loan amounts
and asset values are fixed and that there are no bankruptcy costs,
we show that this problem is equivalent to a linear program.
We develop a duality-based distributed algorithm to solve it which is useful
for applications where it is desirable to avoid centralized data gathering and
computation. Since some applications require forecasting and planning for
a wide variety of different contingencies, we also consider the problem of minimizing
the expectation of the weighted sum of unpaid liabilities under the assumption that
the net external asset holdings of all institutions are stochastic. We show that this
problem is a two-stage stochastic linear program. To solve it, we develop
two algorithms based on Monte Carlo sampling: Benders decomposition algorithm
and projected stochastic gradient descent. We show that if the defaulting nodes
never pay anything, the deterministic optimal cash injection allocation problem is an NP-hard mixed-integer
linear program. However, modern optimization software enables the computation
of very accurate solutions to this problem on a personal computer in a few seconds
for network sizes comparable with the size of the US
banking system. In addition, we address the problem
of allocating the cash injection amount so as to minimize the number of nodes
in default. For this problem, we develop two heuristic algorithms: a reweighted
$\ell_1$ minimization algorithm and a greedy algorithm.
We illustrate these two algorithms using three
synthetic network structures for which the optimal solution can be calculated
exactly. We also compare these two algorithms on three types random networks which
are more complex.
\end{abstract}

\pagestyle{plain} \setcounter{page}{1} \pagenumbering{arabic}
\thispagestyle{empty}

\section{Introduction}
The events of the last several years revealed an acute need for tools to systematically
model, analyze, monitor, and control large financial networks.  Motivated by this need, we propose
to address the problem of optimizing the amount and structure of liquidity assistance in a distressed
financial network, under a variety of modeling assumptions and implementation scenarios.

Two broad applications motivate our work: day-to-day monitoring of financial systems
and decision making during an imminent crisis.  Examples of the latter include the decision in
September 1998 by a group of financial institutions to rescue Long-Term Capital Management,
and the decisions by the Treasury and the Fed in September 2008 to rescue AIG and to let
Lehman Brothers fail. The deliberations leading to these and other similar actions have been
extensively covered in the press.  These reports suggest that the decision making processes
could have benefited from quantitative methods for analyzing potential policies and their likely
outcomes.  In addition, such methods could help avoid systemic crises in the first place,
by informing day-to-day actions of financial institutions, regulators, supervisory authorities,
and legislative bodies.

Given a financial network model, we are interested in addressing the following problem.
\begin{enumerate}
\item[]{\bf Problem I:} Allocate a fixed amount of cash assistance among the nodes in
a financial network in order to minimize the (possibly weighted) sum of unpaid liabilities in the system.
\end{enumerate}
An alternative, Lagrangian, formulation of the same problem, is to both select the amount of
injected cash and determine how to distribute it among the nodes in order to minimize the overall cost
equal to a linear combination of the weighted sum of unpaid liabilities and the amount of injected cash.

We consider a static model with a single maturity date, and with a known network structure. 
At first, we assume that we know both the amounts owed by every node in the network to every other node,
and the net asset amounts available to every node from sources external to the network. Even for
this relatively simple model, Problem~I is far from straightforward, because of a nonlinear relationship
between the cash injection amounts and the loan repayment amounts.  Building upon the results
from~\cite{EiNo01}, we construct algorithms for computing exact solutions for Problem~I and
its Lagrangian variant, by showing in Section~\ref{sec:LP} that both formulations are equivalent
to linear programs under the payment scheme assumed in~\cite{EiNo01}.

Because of this equivalence to linear programs, these problems can be solved for any network topology
by any standard LP solver. In some scenarios, however, this approach may be impractical or undesirable,
as it requires the solver to know the entire network structure, namely, the net external assets of every
institution, as well as the amount owed by each institution to each other institution.
In Section~\ref{sec:DA}, we adapt our framework to avoid centralized data gathering and computation.
We propose a distributed algorithm to solve our linear program. While the algorithm is slower than
standard centralized LP solvers, simulations suggest its practicality for the US banking system.

Some applications, such as stress testing, require forecasting and planning for a wide variety of
different contingencies. Such applications call for the use of stochastic models for the nodes'
external asset values. In this case, we aim to solve a stochastic version of Problem~I.
\begin{enumerate}
\item[]{\bf Problem I-stochastic:} Allocate a fixed amount of cash assistance among the nodes in
a financial network in order to minimize the expectation of the (possibly weighted) sum of unpaid
liabilities in the system.
\end{enumerate}

We prove in Section~\ref{sec:random-e} that, under the payment scheme assumed in~\cite{EiNo01}---whereby
each defaulting bank uses all its available funds to pay all its creditors in proportion to the amounts
it owes them---this problem is equivalent to a stochastic linear program. To solve this problem we develop
two algorithms based on Monte Carlo sampling: a Benders decomposition algorithm described in
Section~\ref{subsec:benders} and a projected stochastic gradient descent algorithm described in
Section~\ref{subsec:stochastic-gd}. Both algorithms are centralized (non-distributed) and assume that
we are able to efficiently obtain independent samples of the external asset vector.

We show in Section~\ref{sec:linear_bankruptcy_costs} that under the all-or-nothing payment scheme
where the defaulting nodes do not pay at all,
Problem I is an
NP-hard mixed-integer linear program. However, we show through simulations that use optimization package
CVX~\cite{cvx,GrBo08} that this problem can be accurately solved in a few seconds on a personal computer
for a network size comparable to the size of the US banking network. 

We also consider another problem where the objective is to minimize the number of defaulting nodes rather
than the weighted sum of unpaid liabilities:
\begin{enumerate}
\item[]{\bf Problem II:} Allocate a fixed amount of cash assistance among the nodes in
a financial network in order to minimize the number of nodes in default.
\end{enumerate}

For Problem~II, we develop two heuristic algorithms: a reweighted $\ell_1$ minimization approach
inspired by~\cite{CaWaBo08} and a greedy algorithm.  We illustrate our algorithms using examples
with synthetic data for which the optimal solution can be calculated exactly. We show through
numerical simulations that the solutions calculated by the reweighted $\ell_1$ algorithm are close
to optimal, and the performance of the greedy algorithm highly depends on the network topology.  
We also compare these two algorithms using three types of random networks for which the optimal
solution is not available.  In one of these three examples the performance of these two algorithms
is statistically indistinguishable; in the second example the greedy algorithm outperforms
the reweighted $\ell_1$ algorithm; and in the third example the reweighted algorithm outperforms
the greedy algorithm.

While Problem~II is unlikely to be of direct practical importance (indeed, it is difficult to imagine
a situation where a regulator would consider the failures of a small local bank and Citi to be equally
bad), it serves as a stepping stone to a more practical and more difficult scenario where the optimization
objective is a linear combination of the weighted unpaid liabilities (as in Problem~I) and the sum
of weights over the defaulted nodes (an extension of Problem~II).
\begin{enumerate}
\item[]{\bf Problem III:} Given a fixed amount of cash to be injected into the system, we consider
an objective function which is a linear combination of the sum of weights over the defaulted nodes
and the weighted sum of unpaid liabilities. 
\end{enumerate}
We show that this problem is equivalent to a mixed-integer linear program.

\subsection{Related Literature}
Contagion in financial networks has been frequently studied in the past, especially after
the financial crisis in 2007-2008. Notable examples of network topology analysis based on real
data are~\cite{BoElSu03,SoBeArGlBe07,CrPe14,VeLe14}. Real data informs the new approaches for assessing
systemic financial stability of banking systems developed
in~\cite{Fu03,ElLeSu05,ElLeSu06,LeLi06,Ha07,DeNg07,PrLeHe08,CaJa09,MaPeEmDe10,GaHaKa11,AnGaKaBrWi12,HaKo13,FoHeSaTa13,AnCrPe14}.

Often, systemic failures are caused by an epidemic of defaults whereby a group of nodes unable to
meet their obligations trigger the insolvency of their lenders, leading to the defaults of lenders'
lenders, etc, until this spread of defaults infects a large part of the system.  For this reason,
many studies have been devoted to discovering network structures conducive to default
contagion~\cite{AlGa00,CiFeSh05,LoBaSc09,GaKa10,AmCoMi10,BlEaKlKlTa11,BaGaGa12}. The relationships between
the probability
of a systemic failure and the average connectivity in the network are investigated
in~\cite{GaKa10,BaGaGa12,AlGa00}. Other features, such as 
as the distribution of degrees and the structure of the subgraphs of contagious links,
are examined in~\cite{AmCoMi10}.

While potentially useful in policymaking, most of these references do not provide specific policy recipes.
One strand of literature on quantitative models for optimizing policy decisions has focused on analyzing
the efficacy of bailouts and understanding the behavior of firms in response to bailouts.  To this end,
game-theoretic models are proposed in~\cite{Mu02} and~\cite{BeTaWe11} that have two agents: the government
and a single private sector entity. The focus of another set of research efforts has been on the setting
of capital and liquidity requirements~\cite{CiFeSh05,HaMa11,GaLeSo12} in order to reduce systemic risk.

Our work contributes to the literature by taking a network-level view of optimal policies and proposing
optimal cash injection strategies for networks in distress.  Our present paper extends our earlier work
reported in~\cite{LiPo12,LiPo13,LiLiPo13}.  In addition to ours, several other papers have
recently considered cash injection policies for lending networks~\cite{De11a,De11b,RoVe11,RoVe13,CaCh13},
all based on the framework proposed in~\cite{EiNo01}.

A cash injection targeting policy is developed in~\cite{De11a,De11b} for an infinitesimally small amount
of injected cash. The basic idea of the policy is to inject the cash into the node with the largest
threat index, which is defined as the derivative of the unpaid liability with respect to the current
asset value. However, extending this idea to construct an optimal cash injection policy for
non-infinitesimal cash amounts produces an inefficient algorithm, as we show in Section~\ref{subsec:demange}.
We show that our own method proposed in Section~\ref{subsec:W} is both more efficient and more flexible,
as it can be extended to a more complex stochastic model described in Section~\ref{sec:random-e}.

In~\cite{RoVe11,RoVe13}, bankruptcy costs are incorporated into the model of~\cite{EiNo01}. The main
contribution of that work is showing that because of the bankruptcy costs, it is sometimes beneficial
for some solvent banks to form bailout consortia and rescue failing banks. However, it may happen that
the solvent banks do not have enough means to effect a bailout, and in this case external intervention
may still be needed.

A multi-period stochastic clearing framework based on~\cite{EiNo01} is proposed in~\cite{CaCh13},
where a lender of last resort monitors the network and may provide liquidity assistance loans to
failing nodes.  The paper proposes several strategies that the lender of last resort might follow
in making its decisions.  One of these strategies, the so-called max-liquidity policy, aims to
solve our Problem I during each period.  However,~\cite{CaCh13} does not describe an algorithm
for solving this problem.

Another related work is~\cite{GlYo14}. Based on the clearing payment framework in~\cite{EiNo01},
the authors of~\cite{GlYo14} study the probability of contagion and amplification of losses due
to network effects when the system suffers a random shock. 

\subsection{Outline of the Paper}
The paper is organized as follows. Section~\ref{sec:model} describes the model of financial networks,
the clearing payment mechanism, and the notation. Section~\ref{sec:LP} shows that 
if each defaulting node pays its creditors in proportion to the owed amounts, then
Problem~I and its
Lagrangian formulation are equivalent to linear programs. A reweighted $\ell_1$
minimization  algorithm to solve Problem~II is also developed in Section~\ref{sec:LP}.
Section~\ref{sec:linear_bankruptcy_costs} considers Problem~I
under the assumption that the defaulting nodes do not pay anything.
We prove that it is then an NP-hard problem and can be formulated
as a mixed-integer linear program that can be efficiently solved using modern optimization software
for network sizes comparable to the size of the US banking system. 
Section~\ref{sec:random-e} extends Problem~I to
the situation where the capital of each institution is random. Section~\ref{sec:DA} proposes
a duality-based distributed algorithm to solve Problem~I and its Lagrangian formulation.

\section{Model and Notation}\label{sec:model}

\begin{table}
\caption{Notation for several vector \vspace*{-0.05in} quantities.}
\begin{center}
\begin{tabular}{|l|p{2in}|}
\hline
\textsc{vector} & \textsc{$i$-th component} \\
\hline
{\bf 0} & 0\\
\hline
{\bf 1} & 1\\
\hline
${\bf e} \geq {\bf 0}$ & net external assets at node $i$ before cash injection \\
\hline
${\bf c} \geq {\bf 0}$ & external cash injection to node $i$ \\
\hline
$\bar{\bf p}$ & the amount node $i$ owes to all its creditors \\
\hline ${\bf p} \leq \bar{\bf p}$ & the total amount node $i$
      actually repays all its creditors on the due date of the loans \\
\hline
$\bar{\bf p} - {\bf p}$ & node $i$'s total unpaid liabilities \\
\hline
${\bf r}$ & remaining cash of node $i$ after clearing payment \\
\hline
${\bf w}$ & the weight of \$1 of unpaid liability at node $i$ \\
\hline
${\bf s}$ & the weight of node $i$'s default \\
\hline
${\bf d}$ & indicator variable of whether node $i$ defaults, i.e., $d_i=1$ if node $i$ defaults; $d_i=0$ otherwise \\
\hline
\end{tabular}
\vspace*{-0.2in}\\
\end{center}
\label{tb:notation}
\end{table}

Our network model is a directed graph with $N$ nodes where a directed edge from node $i$ to node $j$ with weight $L_{ij}>0$ signifies that $i$ owes $\$L_{ij}$ to $j$. This is a one-period model with no dynamics---i.e., we assume that all the loans are due on the same date and all the payments occur on that date. We use the following notation:
\begin{itemize}
\item any inequality whose both sides are vectors is component-wise;
\item $\mathbf{0}$, $\mathbf{1}$, ${\bf e}$, ${\bf c}$, $\bar{\bf p}$, ${\bf p}$, ${\bf r}$, ${\bf w}$, $\bf s$,
and $\bf d$ are all vectors in $\mathbb{R}^N$ defined in Table~\ref{tb:notation};
\item $W = \mathbf{w}^T(\bar{\bf p} - {\bf p})$ is the weighted sum of unpaid liabilities in
the system;
\item $N_d$ is the number of nodes in default, i.e., the number of nodes $i$ whose payments are below their liabilities, $p_i < \bar{p}_i$;
\item $\Pi_{ij}$ is what node $i$ owes to node $j$, as a fraction of the total amount owed by node $i$,
\[
\Pi_{ij} = \left\{\begin{array}{ll} \frac{L_{ij}}{\bar{p}_i} & \mbox{if } \bar{p}_i \neq 0, \\
0 & \mbox{otherwise;} \end{array}\right.
\]
\item $\Pi$ and $L$ are the matrices whose entries are $\Pi_{ij}$ and $L_{ij}$, respectively.
\end{itemize}

Given the above financial system, we consider the proportional payment mechanism and
the all-or-nothing payment mechanism.  The latter can be alternatively interpreted as
the proportional payment mechanism with 100\% bankruptcy costs.
As proposed in~\cite{EiNo01}, the proportional payment mechanism without bankruptcy costs
is defined as follows.\\
\emph{Proportional payment mechanism with no bankruptcy costs:}
\begin{itemize}
\item If $i$'s total funds are at least as large as its liabilities, then all $i$'s creditors get paid in full.
\item If $i$'s total funds are smaller than its liabilities, then $i$ pays all its funds to its creditors.
\item All $i$'s debts have the same seniority.  This means that, if $i$'s liabilities exceed its total funds
then each creditor gets paid in proportion to what it is owed. This guarantees that the amount actually
received by node $j$ from node $i$ is always $\Pi_{ij} p_i$.  Therefore, the total amount received by any
node $i$ from all its borrowers is $\displaystyle \sum_{j=1}^N \Pi_{ji} p_j$.
\end{itemize}
Under these assumptions, a node will pay all the available funds proportionally to its creditors,
up to the amount of its liabilities. The payment vector can lie anywhere in the rectangle
$[\mathbf{0},\bar{\bf p}]$. Under the all-or-nothing payment scenario,
the defaulting nodes do not pay at all.\\
\emph{All-or-nothing payment mechanism:}
\begin{itemize}
\item If $i$'s total funds are at least as large as its liabilities
(i.e., $\displaystyle \sum_{j=1}^N \Pi_{ji} p_j + e_i \geq \bar{p}_i$) then all $i$'s creditors get paid in full.
\item If $i$'s total funds are smaller than its liabilities, then $i$ pays nothing.
\end{itemize}

As defined in~\cite{EiNo01}, a {\em clearing payment vector} ${\bf p}$ is a vector of borrower-to-lender payments that is consistent with the conditions of the payment mechanism. Several algorithms for computing the clearing payment vector are discussed and compared in Appendix~\ref{app:3algo}.

In this paper, we are mostly concerned with Problems~I and~II under the proportional payment scenario
with no bankruptcy costs.
We also prove that the all-or-nothing payment scenario makes Problem~I NP-hard. In this
case, Problem~I can be formulated as a mixed-integer linear program that can be efficiently solved
on a personal computer 
using modern optimization software for network sizes comparable to the size of the US banking system.

\section{Centralized Algorithms for Problems I, II, III under the Proportional Payment Mechanism} \label{sec:LP}

\subsection{Minimizing the Weighted Sum of Unpaid Liabilities is an LP} \label{subsec:W}
Consider a network with a known structure of liabilities $L$ and a known vector ${\bf e}$ of net assets
before cash injection. Using the notation established in the preceding section,
Problem~I seeks a cash injection allocation vector ${\bf c}\geq {\bf 0}$ to minimize the following weighted
sum of unpaid liabilities,
\[
W = {\bf w}^T(\bar{\bf p} - {\bf p}),
\]
subject to the constraint that the total amount of cash injection
does not exceed some given number $C$:
\[
\mathbf{1}^T{\bf c} \leq C.
\]
In this section, we assume proportional payments with no bankruptcy costs.
We first prove that, for any cash injection vector $\bf c$, there exists a unique clearing
payment vector that maximizes the cost $W$.

\begin{lemma} \label{lem:uniqueness}
Given a financial system $(\Pi,\bar{\bf p},{\bf e})$, a cash injection vector $\bf c$ and a weight vector ${\bf w}>\mathbf{0}$, there exists a unique clearing payment vector $\bf p$ minimizing the weighted sum $W={\bf w}^T(\bar{\bf p}-{\bf p})$.
\end{lemma}
\begin{proof}
\emph{Method 1}: First, note that since ${\bf w}$ and $\bar{\bf p}$ do not depend on $\bf p$ or $\bf c$, minimizing $W$ is equivalent to maximizing ${\bf w}^T{\bf p}$. With a fixed cash injection vector $\bf c$, the financial system is equivalent to $(\Pi,\bar{\bf p},{\bf e}+{\bf c})$. Since ${\bf w}>\mathbf{0}$, we have that ${\bf w}^T{\bf p}$ is a strictly increasing function of $\bf p$. By Lemma 4 in \cite{EiNo01}, the clearing payment vector $\bf p$ can be obtained by solving the following linear program:
\begin{align}
& \displaystyle \max_{\bf p} {\bf w}^T{\bf p} \label{eq:payment-vector} \\
& \mbox{subject to } \nonumber\\
& \mathbf{0}\leq {\bf p} \leq \bar{\bf p}, \label{eq:pv-constraint1}\\
& {\bf p} \leq \Pi^T{\bf p} + {\bf e} + {\bf c}. \label{eq:pv-constraint2}
\end{align}
From Theorem 1 in \cite{EiNo01}, there exists a greatest clearing payment vector ${\bf p}^\ast$.
Since $W$ is a strictly increasing function of $\bf p$, ${\bf p}^\ast$ is a solution of
LP~(\ref{eq:payment-vector}-\ref{eq:pv-constraint2}). For any other ${\bf p} \ne {\bf p}^\ast$,
we have $p_i \le p^\ast_i$ for $i=1,2,\cdots,N$ and at least one of these inequalities is strict.
Thus, ${\bf w}^T{\bf p} < {\bf w}^T{\bf p}^\ast$. Therefore ${\bf p}^\ast$ is
the unique solution of LP~(\ref{eq:payment-vector}-\ref{eq:pv-constraint2}).
This completes the proof of Lemma~\ref{lem:uniqueness}.

\emph{Method 2}: Here is another method to prove Lemma~\ref{lem:uniqueness} without using
Theorem~1 in~\cite{EiNo01}. It is clear that LP~(\ref{eq:payment-vector}-\ref{eq:pv-constraint2})
is feasible and bounded so the solution always exists. Assume there are two different
solutions ${\bf p}^1$ and ${\bf p}^2$, and define ${\bf p}^+ \in \mathbb{R}^N$ as
$p^+_i=\displaystyle \max\{p^1_i, p^2_i\}$ for $i=1,2,\cdots,N$.
Then ${\bf w}^T{\bf p}^+ > {\bf w}^T{\bf p}^l$, for $l=1,2$.
Here the inequality is strict because ${\bf p}^1 \ne {\bf p}^2$.

For each $i$, by definition, $p^+_i=p^1_i$ or $p^+_i=p^2_i$. Since ${\bf p}^1$ and ${\bf p}^2$
are both solutions of the LP, they both satisfy constraint~(\ref{eq:pv-constraint1}), i.e.,
$0\leq p^1_i\leq \bar{p}_i$ and $0\leq p^2_i \leq \bar{p}_i$ for all $i$.  Therefore, we also have
$0\leq p^+_i\leq \bar{p}_i$ for all $i$, which means that ${\bf p}^+$ also satisfies
constraint~(\ref{eq:pv-constraint1}).  In addition, both ${\bf p}^1$ and ${\bf p}^2$
satisfy constraint~(\ref{eq:pv-constraint2}), which, along with the fact that all entries
of $\Pi$ are nonnegative, implies that $p^1_i \le \displaystyle \sum_{j=1}^{N}\Pi_{ji}p^+_j+e_i+c_i$
and $p^2_i \le \displaystyle \sum_{j=1}^{N}\Pi_{ji}p^+_j+e_i+c_i$ for all $i$, and therefore also
$p^+_i \le \displaystyle \sum_{j=1}^{N}\Pi_{ji}p^+_j+e_i+c_i$ for all $i$.
Therefore, ${\bf p}^+$ also satisfies constraint~(\ref{eq:pv-constraint2}).
Thus, ${\bf p}^+$ is in the feasible region of LP~(\ref{eq:payment-vector}-\ref{eq:pv-constraint2}) and achieves
a larger value of the objective function than do ${\bf p}^1$ and ${\bf p}^2$, which contradicts
the fact that ${\bf p}^1$ and ${\bf p}^2$ are solutions of LP~(\ref{eq:payment-vector}-\ref{eq:pv-constraint2}).
This completes the proof of Lemma~\ref{lem:uniqueness}.
\end{proof}

We now establish the equivalence of Problem~I and a linear programming problem.

\begin{thm} \label{thm:LP} 
Assume that the liabilities matrix $L$, the asset vector ${\bf e}$, the weight vector $\bf w$,
and the total cash injection amount $C$ are fixed and known. Assume that the system utilizes
the proportional payment mechanism with no bankruptcy costs. Consider Problem~I, i.e.,
the problem of calculating
a cash injection allocation ${\bf c}\geq {\bf 0}$ to minimize the weighted sum of unpaid liabilities
$W={\bf w}^T(\bar{\bf p}-{\bf p})$ subject to the budget constraint ${\bf 1}^T{\bf c} \leq C$.
A solution to this problem can be obtained by solving the following linear program:
\begin{align}
& \displaystyle \max_{{\bf p}, {\bf c}} {\bf w}^T{\bf p}\label{eq:LP}\\
& \mbox{subject to } \nonumber\\
& \mathbf{1}^T{\bf c} \le C, \label{eq:LP-constraint1}\\
& {\bf c} \geq \mathbf{0}, \nonumber\\
& \mathbf{0}\leq {\bf p} \leq \bar{\bf p}, \nonumber\\
& {\bf p} \leq \Pi^T{\bf p} + {\bf e} + {\bf c}. \label{eq:LP-constraint2}
\end{align}
\end{thm}

\begin{proof}
Since the constraints on ${\bf c}$ and ${\bf p}$ in LP~(\ref{eq:LP}-\ref{eq:LP-constraint2})
form a closed and bounded set in $\mathbb{R}^{2N}$, a solution exists. Moreover,
for any fixed ${\bf c}$, it follows from our Lemma~\ref{lem:uniqueness} and Lemma~4 in~\cite{EiNo01}
that the linear program has a unique solution for ${\bf p}$ which is the clearing payment vector for the system.

Let $({\bf p}^\ast,{\bf c}^\ast)$ be a solution to (\ref{eq:LP}-\ref{eq:LP-constraint2}).
Suppose that there exists a cash injection allocation that leads to a smaller cost $W$
than does ${\bf c}^\ast$. In other words, suppose that there exists
${\bf c'}>\mathbf{0}$, with $\mathbf{1}^T{\bf c'} \le C$, such that the corresponding
clearing payment vector ${\bf p'}$ satisfies
$
{\bf w}^T(\bar{\bf p} - {\bf p'}) <
{\bf w}^T(\bar{\bf p} - {\bf p}^\ast),
$
or, equivalently,
\begin{align}
{\bf w}^T{\bf p}^\ast < {\bf w}^T{\bf p'}.
\label{eq:contradiction}
\end{align}
Note that ${\bf c'}$ satisfies the first two constraints of (\ref{eq:LP}-\ref{eq:LP-constraint2}).
Moreover, since ${\bf p'}$ is the corresponding clearing payment vector, the last two constraints
are satisfied as well. The pair $({\bf p'},{\bf c'})$ is thus in the constraint set of
our linear program. Therefore, Eq.~(\ref{eq:contradiction}) contradicts the assumption
that $({\bf p}^\ast,{\bf c}^\ast)$ is a solution to (\ref{eq:LP}-\ref{eq:LP-constraint2}).
This completes the proof that ${\bf c}^\ast$ is the allocation of $C$ that achieves
the smallest possible cost $W$.
\end{proof}

In the Lagrangian formulation of Problem~I, we are given a weight $\lambda$ and must choose the total cash injection amount $C$ and its allocation ${\bf c}$ to minimize $\lambda C+W$. This is equivalent to the following linear program:
\begin{align}
& \displaystyle \max_{C,\mathbf{c},\mathbf{p}} \mathbf{w}^T{\bf p} - \lambda C
\label{eq:LP2}
\\
& \mbox{subject to} \nonumber\\
& \mathbf{1}^T{\bf c} = C, \nonumber\\
& {\bf c} \geq \mathbf{0}, \nonumber\\
& \mathbf{0}\leq {\bf p} \leq \bar{\bf p}, \nonumber\\
& {\bf p} \leq \Pi^T{\bf p} + {\bf e} + {\bf c}. \nonumber
\end{align}
This equivalence follows from Theorem~\ref{thm:LP}: denoting
a solution to (\ref{eq:LP2}) by $(C^\ast,{\bf p}^\ast,{\bf c}^\ast)$,
we see that the pair $({\bf p}^\ast,{\bf c}^\ast)$ must be a solution
to (\ref{eq:LP}-\ref{eq:LP-constraint2}) for $C=C^\ast$.  At the same time, the fact that
$C^\ast$ maximizes the objective function in (\ref{eq:LP2}) means that
it minimizes $\lambda C+ W = \lambda C + \mathbf{w}^T(\bar{\bf p}-{\bf p})$,
since $\bar{\bf p}$ is a fixed constant.

\subsection{Comparison with Demange's Algorithm} \label{subsec:demange}
A cash injection targeting policy is developed in~\cite{De11a,De11b} for an infinitesimally
small amount of the injected cash. The basic idea of Proposition 4 in~\cite{De11b}
is to inject the cash into the node with the largest threat index, which is defined
as the derivative of the sum of the unpaid liabilities with respect to the current asset value.
Moreover, as small amounts of cash are gradually injected, the target remains the same
until at least one bank is fully rescued (i.e., changes from defaulting to solvent),
so the optimal policy for non-infinitesimal amounts of cash would be to keep injecting
cash into the same node until one node changes its state. While no algorithm for the injection
of a non-infinitesimal amount of cash (i.e., for our Problem I) is proposed in~\cite{De11a,De11b},
we construct such an algorithm based on the ideas from~\cite{De11a,De11b}.

{\bf Algorithm for Problem I based on~\cite{De11a,De11b}:}
\begin{enumerate}
\item
Initialization: set cash injection vector ${\bf c} \leftarrow \mathbf{0}$ and
the remaining cash still to be allocated $C_r \leftarrow C$.
\item
Compute the clearing payment vector $\bf p$ for system $(\Pi,\bar{\bf p},{\bf e}+{\bf c})$.
\item
Compute the threat index for system $(\Pi,\bar{\bf p},{\bf e}+{\bf c})$
by solving the linear program (13) in~\cite{De11b}. Select the one with the largest threat
index, denoted as node $i_0$.
\item
Inject a small amount of cash $\delta$ into node $i_0$ and update the clearing
payment vector ${\bf p}'$. Define $\Delta{\bf p}={\bf p}'-{\bf p}$ as the increase of the payment vector
after injecting $\delta$ into node $i_0$.
\item
Compute $\frac{\bar{p}_i-p_i}{\Delta p_i}$ for $i=1,2,...,N$. Select the smallest one, denoted as
node $i_1$. Then node $i_1$ will be the first node that changes from defaulting to being solvent
when we keep injecting cash into node $i_0$.
\item
Set $c_{i_0} \leftarrow c_{i_0} +
\displaystyle \min\left\{C_r,\,\, \delta\frac{\bar{p}_i-p_i}{\Delta p_i}\right\}$.
Set $C_r \leftarrow C_r - \displaystyle \min\left\{C_r,\,\, \delta\frac{\bar{p}_i-p_i}{\Delta p_i}\right\}$.
If $C_r = 0$, stop; otherwise, go to Step~2.
\end{enumerate}
Each iteration of this algorithm computes the clearing payment vector twice: in Steps 2 and 4.
Step 3 moreover involves solving a linear program to obtain the threat index. In the worst case,
the algorithm will stop after $N$ iterations since at each iteration, only one defaulting node
is guaranteed to be rescued. Thus, we would need to solve $N$ LPs and compute the clearing
payment vector $2N$ times in the worst case---much less computationally efficient
than our approach of Theorem~\ref{thm:LP} which requires solving a single LP.
Note that the above algorithm makes a simplifying assumption in Step 4 that a small number $\delta$
can be found in advance such that the injection of $\delta$ in Step 4 does not lead to the rescue
of any banks.  Our algorithm based on Theorem~\ref{thm:LP} does not require this simplifying assumption.
In addition, unlike our LP method, the above algorithm has limited applications.
For example, it is not easy to extend this algorithm to solve Problem~I-stochastic.

\subsection{Minimizing the Number of Defaults} \label{subsec:Nd}

Given that the total amount of cash injection
is $C$, Problem~II seeks to find a cash injection allocation
vector ${\bf c}$ to minimize the number of defaults $N_d$,
i.e., the number of nonzero entries in the vector
$\bar{\bf p} - {\bf p}$.

In this section, we propose two heuristic algorithms to solve
Problem~II approximately.  First,
we adapt the reweighted $\ell_1$ minimization strategy
approach from Section 2.2 of~\cite{CaWaBo08}.  Our algorithm
solves a sequence of weighted versions of the linear program (\ref{eq:LP}-\ref{eq:LP-constraint2}),
with the weights designed to encourage sparsity of $\bar{\bf p} - {\bf p}$.
In the following pseudocode of our algorithm,
${\bf w}^{(m)}$ is the weight vector during the $m$-th iteration.

{\bf Reweighted $\ell_1$ minimization algorithm}:
\begin{enumerate}
\item $m \leftarrow 0$.
\item Select ${\bf w}^0$ (e.g., ${\bf w}^0 \leftarrow {\bf 1}$).
\item Solve linear program (\ref{eq:LP}-\ref{eq:LP-constraint2}) with objective function replaced by
${\bf p}^T{\bf w}^{(m)}$.
\item Update the weights: for each $i=1,\cdots,N$,
\begin{equation}
w_i^{(m+1)}\leftarrow \frac{1}{\exp\left(\bar{p}_i-p^{\ast(m)}_i\right)-1+\epsilon},
\nonumber
\end{equation}
where $\epsilon>0$ is constant, and ${\bf
p}^{\ast(m)}$ is the clearing payment vector obtained in Step 3.
\item If $\|{\bf w}^{(m+1)}-{\bf w}^{(m)}\|_1 < \delta$, where $\delta>0$ is a constant, stop;
else, increment $m$ and go to Step 3.
\end{enumerate}

Note that nodes for which $\bar{p}_i - p_i^{*(m)}$ is very small
require very little additional resources to avoid default.  This is
why Step 4 is designed to give more weight to such nodes, thereby
encouraging larger cash injections into them. On the other hand,
nodes for which $\bar{p}_i - p_i^{*(m)}$ is very large require a lot
of cash to become solvent.  The algorithm essentially ``gives up''
on such nodes by assigning them small weights.

The second heuristic algorithm we develop is a greedy algorithm.
At each iteration of the greedy algorithm, we calculate the clearing payment
vector and select the node with the smallest unpaid liability among all the defaulting
nodes.  We inject cash into that node to rescue it so that during each iteration, we save
the one node that requires the smallest cash expenditure.
In this procedure, we inject the cash sequentially, bailing out some nodes
completely before they fully receive the payments from their borrowers.  These nodes may 
subsequently receive some more cash from their borrowers if their borrowers are rescued
several steps later. Because of this, a rescued node may end up with a surplus.  If this happens,
the node would use its surplus to repay its cash injection.  Such repayments can then
be used to assist other nodes. The algorithm terminates either when there are no defaults
in the system or when
the injected cash reaches the total amount $C$ and no rescued node has a surplus.

{\bf Greedy algorithm}:
\begin{enumerate}
\item $C_r \leftarrow C$, ${\bf c} \leftarrow \mathbf{0}$,
${\bf w} \leftarrow \mathbf{1}$.
\item Solve linear program (\ref{eq:payment-vector}-\ref{eq:pv-constraint2})
to obtain the clearing payment vector ${\bf p}$.
\item Calculate the surplus of each node after clearing:
${\bf r} \leftarrow \Pi^T {\bf p} + {\bf e} + {\bf c} - {\bf p}$.
\item Update the remaining cash to be injected into the system after the rescued
nodes repay their cash injections:
$C_r \leftarrow C_r + \displaystyle \sum_{i=1}^{N} \min\{r_i, c_i\}$,
$c_i \leftarrow c_i - \min\{c_i, r_i\}$ for $i = 1,2,\cdots,N$.
\item If $C_r = 0$ or there are no defaults in the system, stop.
\item Find node $k$ with the minimum unpaid liability $\bar{p}_k - p_k$ among all
defaulting nodes.
\item $c_k \leftarrow \min\{C_r, \bar{p}_k - p_k\}$, $C_r \leftarrow C_r - c_k$,
go to Step 2.
\end{enumerate}

\begin{figure}[t]
    \centering
    \includegraphics[width=0.7\textwidth]{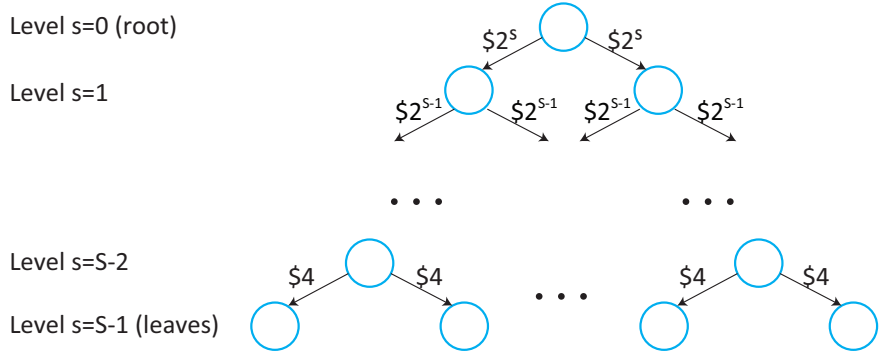}
    \caption{Binary tree network.}
\label{fig:tree}
\end{figure}

\begin{figure}[t]
    \centering
    \includegraphics[width=0.6\textwidth]{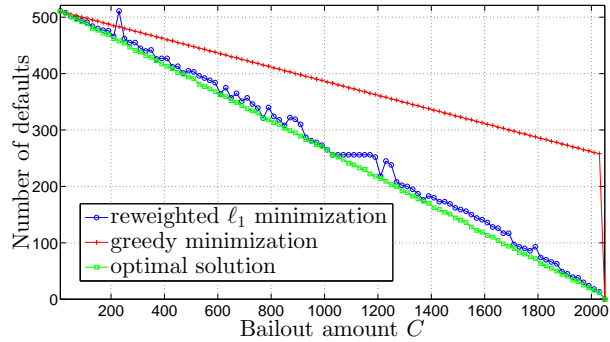}
    \caption{Our algorithm for minimizing the number of defaults vs the optimal solution,
for the binary tree network of Fig.~\ref{fig:tree}.}
\label{fig:evaluation1}
\end{figure}

\subsubsection{Example: A Binary Tree Network}
First, we use a full binary tree with $S$ levels and
$N=2^{S}-1$ nodes. As shown in Fig.~\ref{fig:tree}, levels 0 and
$S-1$ correspond to the root and the leaves, respectively.  Every
node at level $s<S-1$ owes $\$2^{S-s}$ to each of its two creditors
(children).  We set ${\bf e} = {\bf 0}$.

If $C<8$, then all $2^{S-1}-1$ non-leaf nodes are in default, and
the $2^{S-1}$ leaves are not in default.  In aggregate, the nodes at
any level $s<S-1$ owe $\$2^{S+1}$ the nodes at level $s+1$.
Therefore, if $C\geq 2^{S+1}$, then $N_d=0$ can be achieved by
allocating the entire amount to the root node.

For $8\leq C < 2^{S+1}$, we first observe that if $C = 2^{S+1-s}$ for
some integer $s$, then the optimal solution is to allocate the
entire amount to a node at level~$s$.  This would prevent the
defaults of this node and all its $2^{S-s-1}-2$ non-leaf descendants, leading to
$2^{S-1}-2^{S-s-1}$ defaults.  If $C$ is not a power of two, we can
represent it as a sum of powers of two and apply the same argument
recursively, to yield the following optimal number of defaults:
\[
N_d=\displaystyle T(S)-\sum_{u=4}^U b(u)\cdot T(u-2),
\]
where $T(x)=2^{x-1}-1$ is the number of non-leaf nodes in an
$x$-level complete binary tree, $b(u)$ is the $u$-th bit in the binary
representation of $C$ (right to left) and $U$ is the number of bits.
To summarize, the smallest number of defaults $N_d$, as a function
of the cash injection amount $C$, is:
\begin{equation}
N_d(C) \! = \! \left\{\!\!\!\! \begin{array}{ll}
T(S) & \!\!\! \mbox{if } C < 8, \\
\displaystyle T(S)\! - \!
\sum_{u=4}^U b(u)T(u-2) & \!\!\! \mbox{if } 8\leq C < 2^{S+1}, \\
0 & \!\!\! \mbox{if } C \geq 2^{S+1}.
\end{array}\right.
\label{eq:Nd1}
\end{equation}

In our test, we set $S=10$. The green line in
Fig.~\ref{fig:evaluation1} is a plot of the minimum number of
defaults as a function of $C$ from Eq.~(\ref{eq:Nd1}).  The blue line is the solution
calculated by the reweighted $\ell_1$ minimization algorithm with
$\epsilon=0.001$ and $\delta=10^{-6}$. The algorithm was run
using six different initializations: five random ones and $\bf
w^{(0)} = 1$.  Among the six solutions, the one with the smallest
number of defaults was selected.  The red line is the solution
calculated by the greedy algorithm.  As evident from
Fig.~\ref{fig:evaluation1}, the results of the reweighted $\ell_1$ minimization
algorithm are very close to the optimal for the entire range of $C$.
The performance of the greedy algorithm is poor.  The
greedy algorithm always injects cash into the nodes at level $S-2$ which have
the smallest unpaid liabilities. For $C \geq 16$, this strategy is inefficient since
spending \$16 on a node at level $S-3$ rescues both that node and its two children,
whereas spending \$16 on two nodes at level $S-2$ only rescues those two nodes.


\begin{figure}[t]
    \centering
    \includegraphics[width=0.6\textwidth]{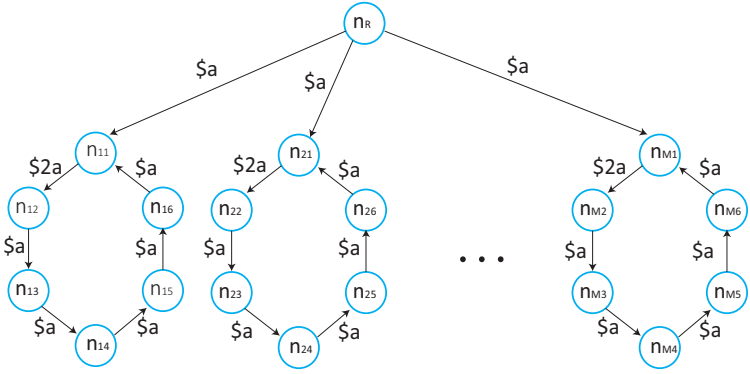}
    \caption{Network topology with cycles.}
\label{fig:cycle}
\end{figure}

\begin{figure}[t]
    \centering
    \includegraphics[width=0.7\textwidth]{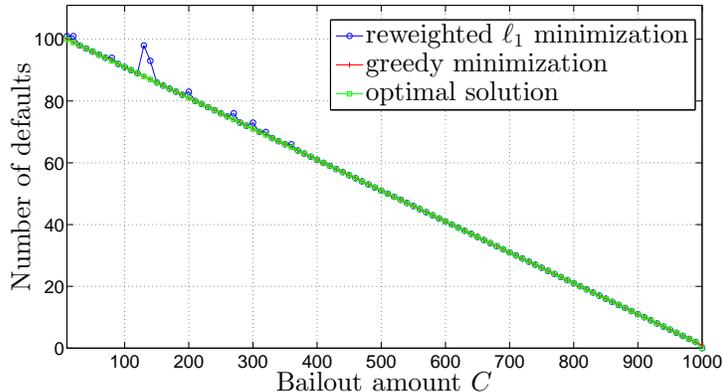}
    \caption{Our algorithm for minimizing the number of defaults vs the optimal solution,
for the network of Fig.~\ref{fig:cycle}.} \label{fig:evaluation2}
\end{figure}

\subsubsection{Example: A Network with Cycles}
Second, we test our algorithms on a network with cycles shown in
Fig.~\ref{fig:cycle}.  The network contains $M$ cycles with six
nodes each. The nodes in the $k$-th cycle are denoted $n_{k1},
n_{k2},\cdots,n_{k6}$. Node $n_{k1}$ owes $\$2a$ to $n_{k2}$. Node
$n_{k6}$ owes $\$a$ to $n_{k1}$. For $i=2,\cdots,5$, $n_{ki}$ owes
$\$a$ to $n_{k(i+1)}$. The root node, denoted as $n_R$, owes $\$a$
to $n_{k1}$, for every $k=1,2,\cdots,M$. We set ${\bf e}=\mathbf{0}$.

If $C < a$, then the root node and all $M$ nodes connected to the
root, $n_{k1} (k=1,2,\cdots,M$), are in default. The remaining $5M$
nodes are not in default.

If $C \geq aM$, then allocating the entire amount $C$ to the root yields
zero defaults.

If $a \leq C <aM$, then giving $\$a$ to node $n_{k1}$ will prevent it from
defaulting. Thus, the total number of defaults in this case is $M+1-[C/a]$.

Summarizing, for this network structure, the smallest number of defaults $N_d$,
as a function of the cash injection amount $C$, is:
\begin{equation}
N_d(C) = \left\{\begin{array}{ll}
M+1 & \mbox{if } C < a, \\
M+1-[C/a] & \mbox{if } a\leq C < aM, \\
0 & \mbox{if } C \geq aM.
\end{array}\right.
\label{eq:Nd2}
\end{equation}

In our test, we set $a=10$ and $M=100$. In Fig.~\ref{fig:evaluation2},
the green line is a plot of the minimum number of defaults as a function of $C$.
The blue line is the solution
calculated by the reweighted $\ell_1$ minimization algorithm with
$\epsilon=0.001$ and $\delta=10^{-6}$. The algorithm was run
using six different initializations: five random ones and
$\bf w^{(0)} = 1$.  Among the six solutions, the one with the smallest
number of defaults was selected.  The red line is the solution
calculated by the greedy algorithm.  As evident from
Fig.~\ref{fig:evaluation2}, the results produced by both algorithms
are very close to the optimal ones.
The greedy algorithm achieves the optimal for the entire range of $C$
except the point $C = 1000$. When $C = 1000$, the optimal 
strategy is to inject $\$1000$ into the root node whereas the greedy
algorithm injects $\$10$ into $n_{k1}$ for $k = 1,2,\cdots,100$.

\begin{figure}[t]
    \centering
    \includegraphics[width=0.6\textwidth]{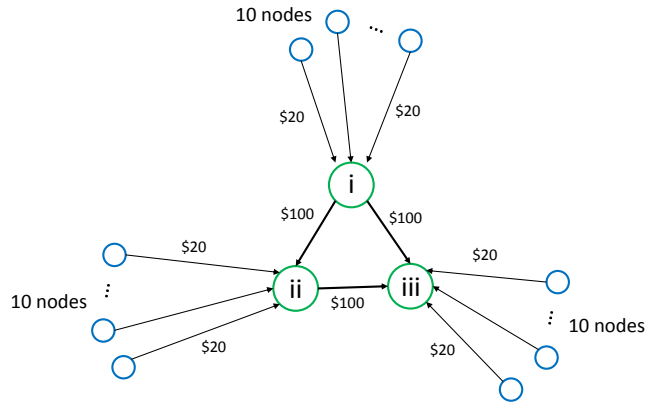}
    \caption{Core-periphery network topology.}
\label{fig:cp}
\end{figure}

\begin{figure}[t]
    \centering
    \includegraphics[width=0.7\textwidth]{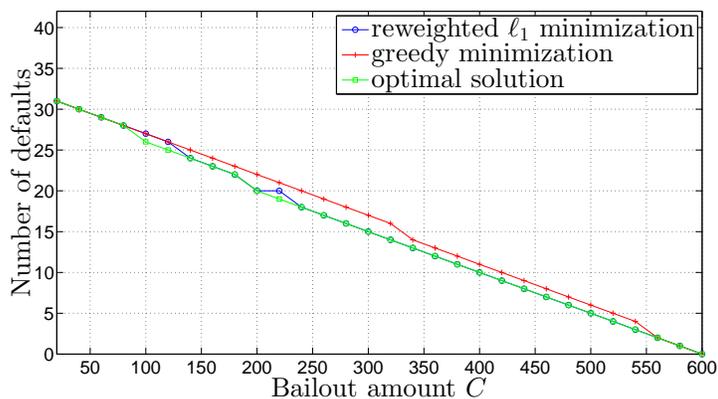}
    \caption{Our algorithm for minimizing the number of defaults vs the optimal solution,
for the network of Fig.~\ref{fig:cp}.} \label{fig:evaluation3}
\end{figure}

\subsubsection{Example: A Core-Periphery Network}
Third, we test our algorithm on a simple core-periphery network, since core-periphery models are widely used to model banking systems~\cite{ElGoJa13,Fa14,VeLe14,CrPe14}. In Fig.~\ref{fig:cp}, i, ii, and iii are the three core nodes. Node i owes \$100 each to nodes ii and iii, and node ii owes \$100 to iii. Ten periphery nodes are attached to each core node, and each periphery node owes \$20 to its core node. There are no external assets in the system and therefore in the absence of an external injection of cash, all the nodes are in default except node iii.

If the cash injection amount is $C < 100$, the optimal solution is to select any $[C/20]$ periphery nodes and give \$20 to each of them. This reduces the number of defaults by $[C/20]$.

If $100 \le C < 200$, we first select any five periphery nodes of core node ii and give \$20 to each of them, because this saves both node ii and these five periphery nodes. Then we select any other $[(C-100)/20]$ periphery nodes and give \$20 to each. This decreases the number of defaults by $[C/20]+1$.

If $200 \le C < 600$, we first use \$200 to rescue all 10 periphery nodes of core node i, saving i, ii, and these 10 periphery nodes. Then we select any other $[(C-200)/20]$ periphery nodes and give \$20 to each. This decreases the number of defaults by $[C/20]+2$.

If $C \geq 600$, then all the nodes can be rescued by giving \$20 to each periphery node.

To sum up, for this core-periphery network structure, the smallest number of defaults $N_d$,
as a function of the cash injection amount $C$, is:
\begin{equation}
N_d(C) = \left\{\begin{array}{ll}
32-[C/20] & \mbox{if } C< 100, \\
31-[C/20] & \mbox{if } 100 \le C < 200, \\
30-[C/20] & \mbox{if } 200 \le C < 600, \\
0 & \mbox{if } C \ge 600.
\end{array}\right.
\label{eq:Nd3}
\end{equation}

In Fig.~\ref{fig:evaluation3}, the green line is a plot of this minimum
number of defaults as a function of $C$. 
The blue line is the solution
calculated by our reweighted $\ell_1$ minimization algorithm with
$\epsilon=0.001$ and $\delta=10^{-6}$. The algorithm was run
using six different initializations: five random ones and
$\bf w^{(0)} = 1$.  Among the six solutions, the one with the smallest
number of defaults was selected.  The red line is the solution
calculated by the greedy algorithm.  As evident from
Fig.~\ref{fig:evaluation3}, the results produced by the reweighted $\ell_1$
algorithm are very close to the optimal ones for the entire range of $C$.
Note that for the greedy algorithm, the performance depends on the order of rescuing
nodes with the same unpaid liability amounts. For example, if the greedy
algorithms rescue the periphery nodes of core node iii first, the performance
would be poor.

\subsubsection{Example: Three Random Networks}
We now compare the reweighted $\ell_1$ minimization algorithm to the greedy algorithm
using more complex
network topologies in which the optimal solution is difficult to calculate directly.

\begin{figure}[t]
    \centering
    \includegraphics[width=0.6\textwidth]{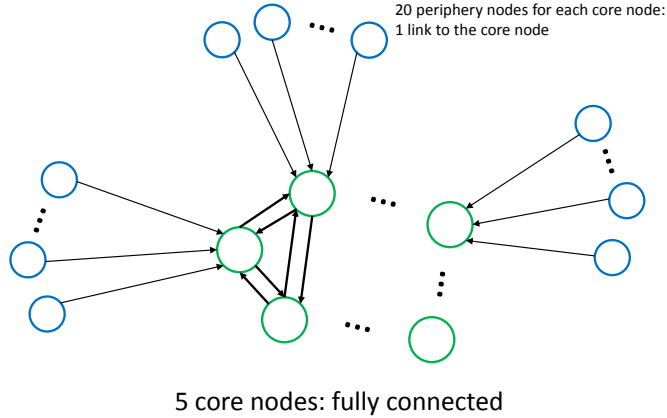}
    \caption{Random core-periphery network to compare the reweighted $\ell_1$ algorithm
    and the greedy algorithm.}
\label{fig:random_core_periphery}
\end{figure}

\begin{figure}[t]
    \centering
    \includegraphics[width=0.6\textwidth]{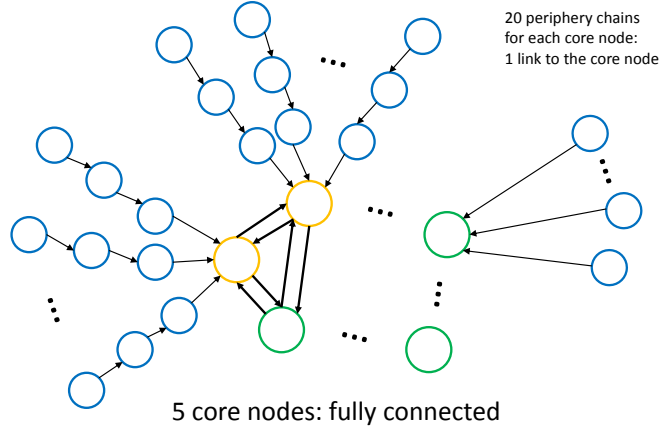}
    \caption{Random core-periphery network with long chains to compare the reweighted
    $\ell_1$ algorithm and the greedy algorithm.}
\label{fig:random_core_periphery_chains}
\end{figure}

We construct three types of random networks, all having external asset vector ${\bf e}=\mathbf{0}$.
The first one is a random graph with
30 nodes. For any pair of nodes $i$ and $j$, $L_{ij}$ is zero with probability 0.8 
and is uniformly distributed in $[0,2]$ with probability 0.2.

The second one is a random core-periphery network which is illustrated in
Fig.~\ref{fig:random_core_periphery}. The core contains five nodes
which are fully connected. The liability from one core node to every other core node
is uniformly distributed in $[0,20]$. Each core node has 20 periphery nodes. Each
periphery node owes money only to its core node.  This amount of money is uniformly
distributed in $[0,1]$.

The third one is a random core-periphery network with chains of periphery nodes.
As shown in Fig.~\ref{fig:random_core_periphery_chains},
the core contains five nodes which are fully connected.
The liability from one core node to every other core node is uniformly
distributed in [0,20]. Each core node has 20 periphery chains connected to it, each
chain consisting of either a single periphery node (short chains) or 3 periphery nodes
(long chains).  Each core node has either only short periphery chains connected to it
or only long periphery chains connected to it.  There are two core nodes with 
long periphery chains. The liability amounts along each long chain are the same,
and are uniformly distributed in [0,1].  The liability amounts along each short chain
are also uniformly distributed in [0,1].

For each of these three random networks, we generate 100 samples from the distribution
and run both the reweighted $\ell_1$ minimization algorithm and the greedy algorithm
on each sample network.  In the reweighted $\ell_1$ minimization algorithm,
we set $\epsilon = 0.001$, $\delta = 10^{-6}$. We run the algorithm
using six different initializations: five random ones and
$\bf w^{(0)} = 1$.  Among the six solutions, the one with the smallest
number of defaults is selected.

The results are shown in Figs.~\ref{fig:evaluation4},~\ref{fig:evaluation5}, and~\ref{fig:evaluation6}.
The blue and red solid lines
represent the average numbers of defaulting nodes after the cash injection allocated
by the two algorithms: blue for the reweighted $\ell_1$ minimization and red for the
greedy algorithm.  The dashed lines show the error bars for the estimates of the average.
Each error bar is $\pm$two standard errors.

From Fig.~\ref{fig:evaluation4}, we see the performance of the reweighted $\ell_1$ algorithm 
is close to the greedy algorithm on the random networks.
From Fig.~\ref{fig:evaluation5} and Fig.~\ref{fig:evaluation6}, we see that on random
core-periphery networks, the greedy algorithm performs better than the reweighted $\ell_1$
algorithm, while on random core-periphery networks with chains, the reweighted $\ell_1$
algorithm is better.

\begin{figure}
    \centering
    \includegraphics[width=0.7\textwidth]{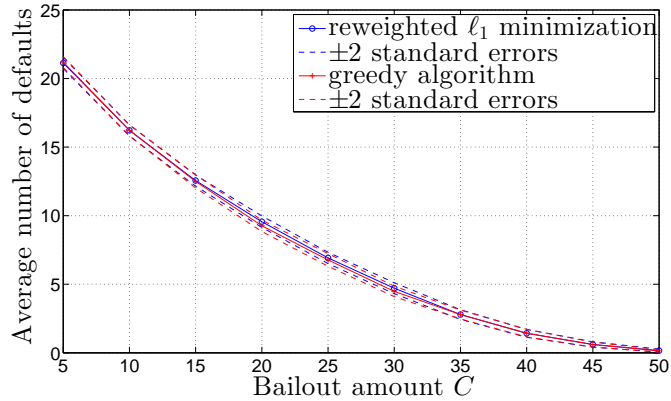}
    \caption{Two heuristic algorithms for minimizing the number of defaults: evaluation on
 random networks.} \label{fig:evaluation4}
\end{figure}

\begin{figure}
    \centering
    \includegraphics[width=0.7\textwidth]{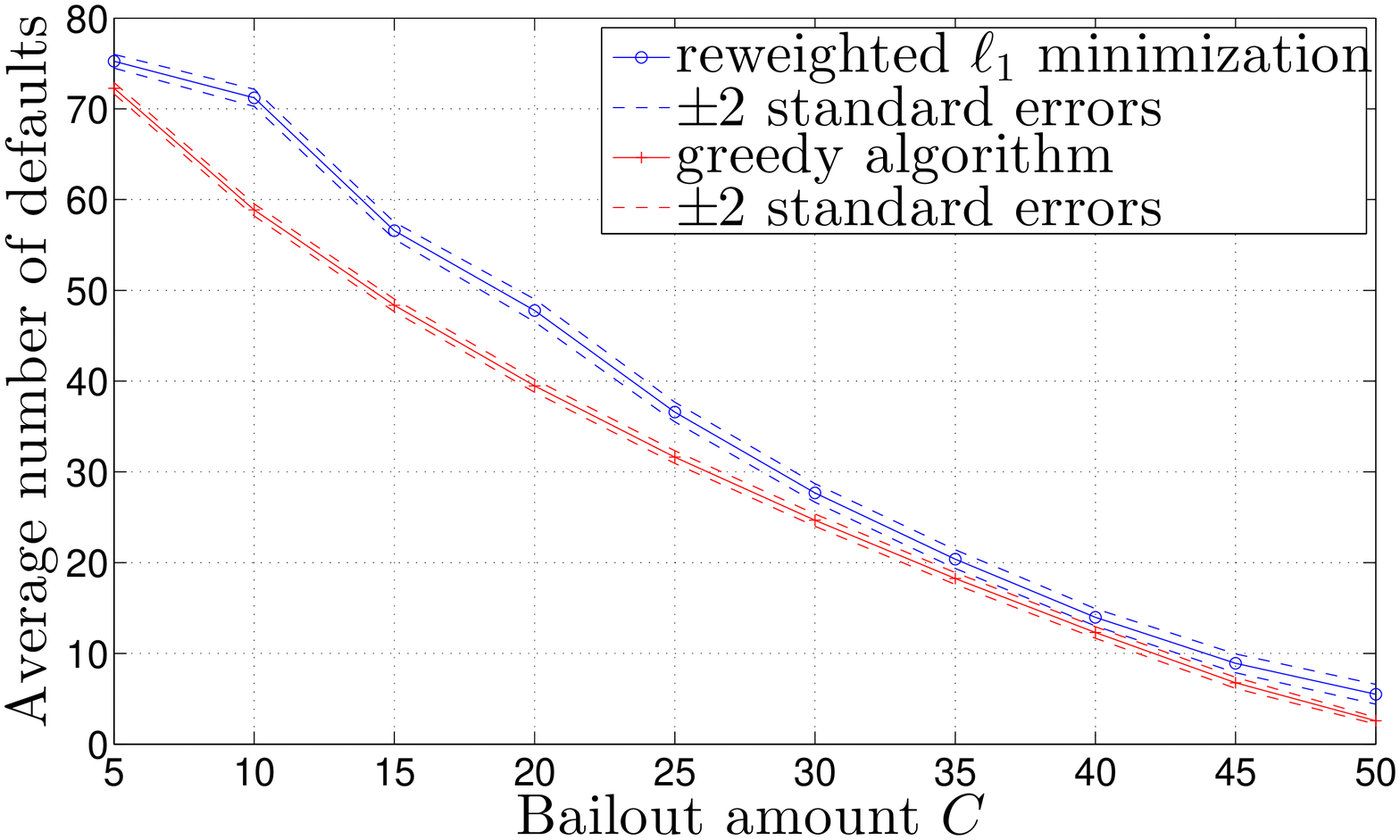}
    \caption{Two heuristic algorithms for minimizing the number of defaults: evaluation on
random core-periphery networks.} \label{fig:evaluation5}
\end{figure}

\begin{figure}
    \centering
    \includegraphics[width=0.7\textwidth]{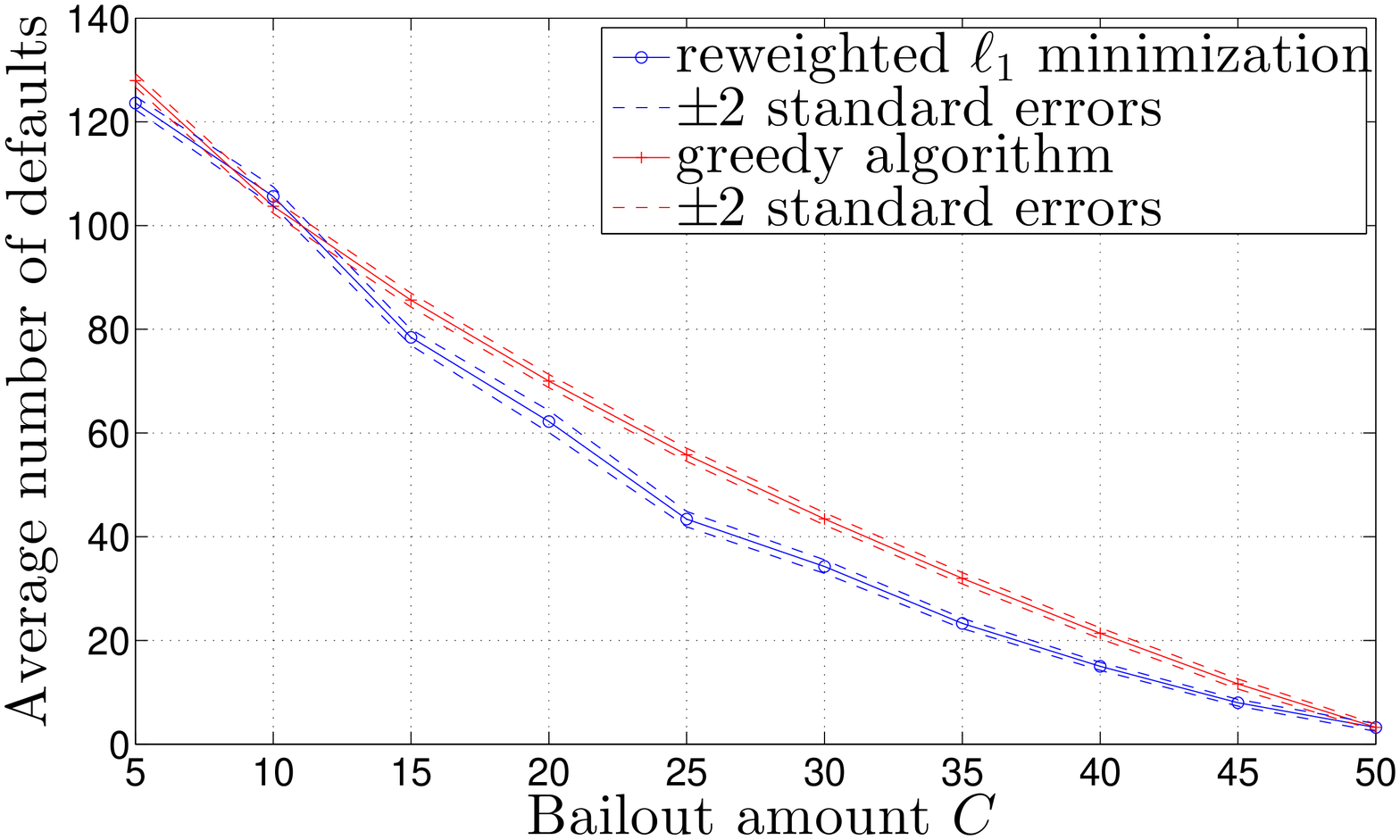}
    \caption{Two heuristic algorithms for minimizing the number of defaults: evaluation on
random core-periphery networks with long chains.} \label{fig:evaluation6}
\end{figure}

\subsection{Minimizing a Linear Combination of Weighted Unpaid Liabilities and Sum
of Weights over Defaulting Nodes} \label{subsec:comb}

We now investigate Problem~III which is a combination of Problem~I and Problem~II.
Instead of just minimizing the weighted sum of unpaid liabilities or the number
of defaulting nodes, we consider an objective function which is a linear combination
of the sum of weights over the defaulted nodes and the weighted sum of unpaid liabilities:
\[
D = {\bf w}^T(\bar{\bf p} - {\bf p}) + {\bf s}^T{\bf d},
\]
As defined in Table~\ref{tb:notation}, $d_i=\mathbb{I}_{\bar{p}_i-p_i>0}$ is a binary variable indicating whether node $i$ defaults, and $s_i$ is the weight of node $i$'s default.

Since $D$ is strictly decreasing with respect to $\bf p$, Lemma 4 in \cite{EiNo01} implies
that minimizing $D$ will yield a clearing payment vector. In light of this fact, we prove that minimizing $D$ subject to a fixed injected cash amount $C$ is equivalent to a mixed-integer linear program.

\begin{thm} \label{thm:MILP} 
Assume that the liabilities matrix $L$, the external asset vector ${\bf e}$, the weight vectors $\bf w > 0$ and $\bf s > 0$ and the total cash injection amount $C$ are fixed and known. Assume that the system utilizes the proportional payment mechanism with no bankruptcy costs. Define $\bf d$ as in Table~\ref{tb:notation}. Then the optimal cash allocation policy to minimize the cost function $D = {\bf w}^T(\bar{\bf p} - {\bf p}) + {\bf s}^T{\bf d}$ can be obtained by solving the following mixed-integer linear program:
\begin{align}
& \displaystyle \max_{{\bf p}, {\bf c}, {\bf d}} {\bf w}^T{\bf p}-{\bf s}^T{\bf d}\label{eq:MILP}\\
& \mbox{subject to } \nonumber\\
& \mathbf{1}^T{\bf c} \le C, \label{eq:MILP-c5}\\
& {\bf c} \geq \mathbf{0}, \label{eq:MILP-c6}\\
& \mathbf{0}\leq {\bf p} \leq \bar{\bf p}, \label{eq:MILP-c1}\\
& {\bf p} \leq \Pi^T{\bf p} + {\bf e} + {\bf c}, \label{eq:MILP-c2}\\
& \bar{p}_i-p_i \leq \bar{p}_i d_i \mbox{, for } i=1,2,\cdots,N, \label{eq:MILP-c3}\\
& d_i \in \{0,1\} \mbox{, for } i=1,2,\cdots,N. \label{eq:MILP-c4}
\end{align}
\end{thm} 

\begin{proof}
Let (${\bf p}^\ast$, ${\bf c}^\ast$, ${\bf d}^\ast$) be a solution of the mixed-integer
linear program (\ref{eq:MILP}--\ref{eq:MILP-c4}).
We first show that ${\bf p}^\ast$ is a clearing payment vector, i.e., that for each $i$,
we have $p^\ast_i=\bar{p}_i$ or
$p^\ast_i=\displaystyle \sum_{j=1}^{N} \Pi_{ji}p^\ast_j + e_i+c_i$. Assume that this is
not the case for some node $k$, i.e., that $p^\ast_k<\bar{p}_k$ and
$p^\ast_k<\displaystyle \sum_{j=1}^{N} \Pi_{jk}p^\ast_j + e_k + c_k$. We construct
a vector ${\bf p}^\epsilon$ which is equal to ${\bf p}^\ast$ in all components except the
$k$-th component.  We set the $k$-th component of ${\bf p}^\epsilon$ to be
$p^\epsilon_k=p^\ast_k+\epsilon$, where $\epsilon>0$ is small enough to
ensure that $p^\epsilon_k<\bar{p}_k$ and
$p^\epsilon_k<\displaystyle \sum_{j=1}^{N} \Pi_{jk}p^\epsilon_j + e_k + c_k$.
Since $\Pi$ is a matrix with non-negative entries, for any $i \ne k$, we have:
\begin{equation}
p^\epsilon_i = p^\ast_i<\displaystyle \sum_{j=1}^{N} \Pi_{ji}p^\ast_j + e_i + c_i < \sum_{j=1}^{N} \Pi_{ji}p^\epsilon_j + e_i + c_i. \nonumber
\end{equation}
In addition, $\bar{p}_k-p^\epsilon_k < \bar{p}_k-p^\ast_k \le \bar{p}_k d_k$.
Thus, (${\bf p}^\epsilon$, ${\bf c}^\ast$, ${\bf d}^\ast$) is also in the feasible
region of (\ref{eq:MILP}--\ref{eq:MILP-c4}) and achieves a larger value of the objective
function
than (${\bf p}^\ast$, ${\bf c}^\ast$, ${\bf d}^\ast$). This contradicts the fact
that (${\bf p}^\ast$, ${\bf c}^\ast$, ${\bf d}^\ast$) is a solution of
(\ref{eq:MILP}--\ref{eq:MILP-c4}).
Hence, ${\bf p}^\ast$ is a clearing payment vector.

Second, we show that $d^\ast_i=\mathbb{I}_{\bar{p}_i-p^\ast_i>0}$.
If $\bar{p}_i-p^\ast_i>0$, then $d^\ast_i=1$ due to constraints~(\ref{eq:MILP-c3})
and~(\ref{eq:MILP-c4}). If $\bar{p}_i-p^\ast_i=0$, then constraint~(\ref{eq:MILP-c3})
is always true. In this case the fact that $s_i > 0$ implies that, in order to
maximize the objective function, $d^\ast_i$ must be zero.
Thus, $d^\ast_i=\mathbb{I}_{\bar{p}_i-p^\ast_i>0}$.

So far, we have proved that ${\bf p}^\ast$ and ${\bf d}^\ast$ are the clearing
payment vector and default indicator vector, respectively, for cash injection vector
${\bf c}^\ast$. We now prove by contradiction that ${\bf c}^\ast$ is the optimal cash
injection allocation. Assume ${\bf c}' \ne {\bf c}^\ast$ leads to a strictly smaller
value of the cost function $D$ than does ${\bf c}^\ast$. In other words, suppose that
${\bf c}'$ satisfies the constraints (\ref{eq:MILP-c5}) and (\ref{eq:MILP-c6}), and
that the corresponding clearing payment vector ${\bf p}'$ and default indicator
vector ${\bf d}'$ satisfy
${\bf w}^T(\bar{\bf p} - {\bf p}') + {\bf s}^T{\bf d}' < {\bf w}^T(\bar{\bf p} - {\bf p}^\ast) + {\bf s}^T{\bf d}^\ast$, which is equivalent to:
\begin{align}
{\bf w}^T{\bf p}'-{\bf s}^T{\bf d}' > {\bf w}^T{\bf p}^\ast-{\bf s}^T{\bf d}^\ast.
\nonumber
\end{align}
Since ${\bf p}'$ is the corresponding clearing payment vector, constraint (\ref{eq:MILP-c1}) and (\ref{eq:MILP-c2}) are satisfied. Moreover, ${\bf d}'$ is the corresponding default indicator vector satisfying constraint (\ref{eq:MILP-c3}) and (\ref{eq:MILP-c4}) for ${\bf c}'$. So (${\bf c}'$,${\bf p}'$,${\bf d}'$) is in the feasible region of
(\ref{eq:MILP}--\ref{eq:MILP-c4}) and achieves a larger objective function than (${\bf p}^\ast$, ${\bf c}^\ast$, ${\bf d}^\ast$), which contradicts the fact that (${\bf p}^\ast$, ${\bf c}^\ast$, ${\bf d}^\ast$) is the solution of (\ref{eq:MILP}--\ref{eq:MILP-c4}).
\end{proof}

\begin{figure}
    \centering
    \includegraphics[width=0.5\textwidth]{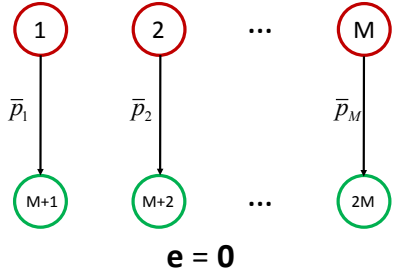}
    \caption{Financial network used in the proof of Theorem~\ref{thm:bankruptcy_costs_NP}. For this network,
  Problem~I under the all-or-nothing payment mechanism is a knapsack problem.} \label{fig:np}
\end{figure}

\section{All-or-Nothing Payment Mechanism}
\label{sec:linear_bankruptcy_costs}
We now show that under the all-or-nothing payment mechanism, Problem~I is NP-hard.
Despite this fact, we show through simulations that for network sizes comparable
to the size of the US banking system, this problem can be solved in a few seconds
on a personal computer using modern optimization software.

\begin{thm} \label{thm:bankruptcy_costs_NP}
With the all-or-nothing payment mechanism, Problem~I can be reduced to a knapsack problem,
which means that Problem~I is NP-hard.
\end{thm}

\begin{proof}
Consider the network depicted in Fig.~\ref{fig:np}.  The network has $N=2M$ nodes
where $M$ is a positive integer.  We let $L_{i,M+i}=\bar{p}_i$ for $i=1,2,\cdots,M$;
for all other pairs $(i,j)$, we set $L_{ij}=0$. We set the external asset vector to zero:
${\bf e}=\mathbf{0}$. We set all the weights to 1: ${\bf w} = {\bf 1}$.
We let $x_i$ be the rescue indicator variable for node $i$, i.e., $x_i = 0$ if $i$ is
in default and $x_i=1$ if $i$ is fully rescued, for $i=1,\cdots,M$.

Note that under the all-or-nothing payment mechanism, fully rescuing node $i$ for any
$i=1,\cdots,M$ in Fig.~\ref{fig:np} means injecting $c_i = \bar{p}_i$.  On the other hand, injecting any
other nonzero amount $c_i < \bar{p}_i$ is wasteful, as it does not reduce the total
amount of unpaid liabilities in the system.  Therefore, for each defaulting node $i$
we have $x_i = 0$, $c_i = 0$, and $p_i = 0$, and for each rescued node $i$ we have
$x_i = 1$, $c_i = \bar{p}_i$, and $p_i = \bar{p}_i$.
The reduction in the total amount of unpaid obligations due to the cash injection is
\[
\sum_{i=1}^M x_i\bar{p}_i.
\]
We must select ${\bf x}$ to maximize this amount, subject to the budget constraint 
$\displaystyle \sum_{i=1}^M x_i \bar{p}_i \leq C$
that says that the total amount of cash injection spent on fully rescued nodes must not exceed $C$:
\begin{align}
& \displaystyle \max_{\bf x} \sum_{i=1}^{M} x_i \bar{p}_i\label{eq:Knapsack}\\
& \mbox{subject to } \nonumber\\
& \sum_{i=1}^{M} x_i \bar{p}_i \le C, \nonumber\\
& x_i \in \{0,1\} \mbox{, for } i=1,2,\cdots,M. \nonumber
\end{align}
If any cash remains, it can be arbitrarily allocated among the remaining nodes or not spent
at all, because partially rescuing a node does not lead to any improvement of the objective function.
Program~(\ref{eq:Knapsack}) is a {\em knapsack problem},
a well-known NP-hard problem. Thus,
Problem~I under the all-or-nothing payment mechanism for the network of Fig.~\ref{fig:np},
which can be reduced to~(\ref{eq:Knapsack}), is an NP-hard problem.
\end{proof}

We now establish a mixed-integer linear program to solve Problem~I with the all-or-nothing payment mechanism.

\begin{thm} \label{thm:MILP2}
Assume that the liabilities matrix $L$, the external asset vector ${\bf e}$, the weight vector
${\bf w}>\mathbf{0}$ and the total cash injection amount $C$ are fixed and known.
Assume the all-or-nothing payment mechanism.
Then Problem~I is equivalent to the following mixed-integer linear program:
\begin{align}
& \displaystyle \max_{{\bf p}, {\bf c}, {\bf d}} {\bf w}^T{\bf p}\label{eq:MILP2}\\
& \mbox{subject to } \nonumber\\
& \mathbf{1}^T{\bf c} \le C, \label{eq:MILP2-c1}\\
& {\bf c} \geq \mathbf{0}, \label{eq:MILP2-c2}\\
& p_i = \bar{p}_i(1-d_i) \mbox{, for } i=1,2,\cdots,N, \label{eq:MILP2-c3}\\
& \bar{p}_i-\sum_{j=1}^{N}\Pi_{ji}p_j-e_i-c_i \le \bar{p}_i d_i \mbox{, for } i=1,2,\cdots,N, \label{eq:MILP2-c4}\\
& d_i \in \{0,1\} \mbox{, for } i=1,2,\cdots,N. \label{eq:MILP2-c5}
\end{align}
\end{thm}

\begin{proof}
Let (${\bf p}^\ast,{\bf c}^\ast,{\bf d}^\ast$) be a solution of the mixed-integer
linear program (\ref{eq:MILP2}--\ref{eq:MILP2-c5}). We first show that ${\bf p}^\ast$
is the clearing payment vector corresponding to ${\bf c}^\ast$. For node $i$,
if $\bar{p}_i>\displaystyle \sum_{j=1}^{N}\Pi_{ji}p^\ast_j+e_i+c_i$, then from
constraints~(\ref{eq:MILP2-c4}) and~(\ref{eq:MILP2-c5}) it follows that
$d^\ast_i=1$ so that $p^\ast_i=0$.
If $\bar{p}_i\le\displaystyle \sum_{j=1}^{N}\Pi_{ji}p^\ast_j+e_i+c_i$,
then constraint~(\ref{eq:MILP2-c4}) is satisfied for both $d_i=0$ and $d_i=1$.
In this case, in order to maximize the objective function, it must be that
$d^\ast=0$ and $p^\ast_i=\bar{p}_i$.
This completes the proof that ${\bf p}^\ast$ is the clearing payment vector
corresponding to ${\bf c}^\ast$ under the all-or-nothing payment mechanism.

Second, we prove by contradiction that ${\bf c}^\ast$ is the optimal allocation.
Assume that ${\bf c}'$ leads to a smaller weighted sum of unpaid liabilities,
or equivalently, a larger value of ${\bf w}^T{\bf p}'$, where ${\bf p}'$ is
the clearing payment vector corresponding to ${\bf c}'$. Since ${\bf p}'$ is
a clearing payment vector, we have that if $\bar{p}_i>\displaystyle \sum_{j=1}^{N}\Pi_{ji}p'_j+e_i+c'_i$
then $p'_i=0$;
and if $\bar{p}_i\le \displaystyle \sum_{j=1}^{N}\Pi_{ji}p'_j+e_i+c'_i$ then
$p'_i=\bar{p}_i$. We define vector ${\bf d}'$ as $d'_i=0$ for $p'_i=\bar{p}_i$
and $d'_i=1$ otherwise. Then (${\bf p}',{\bf c}',{\bf d}'$) is located in
the feasible region of MILP (\ref{eq:MILP2}--\ref{eq:MILP2-c5}) but leads to
a larger value of the objective function than (${\bf p}^\ast,{\bf c}^\ast,{\bf d}^\ast$).
This contradicts the fact that (${\bf p}^\ast,{\bf c}^\ast,{\bf d}^\ast$)
is a solution of (\ref{eq:MILP2}--\ref{eq:MILP2-c5}).
\end{proof}

\begin{figure}
    \centering
    \includegraphics[width=0.7\textwidth]{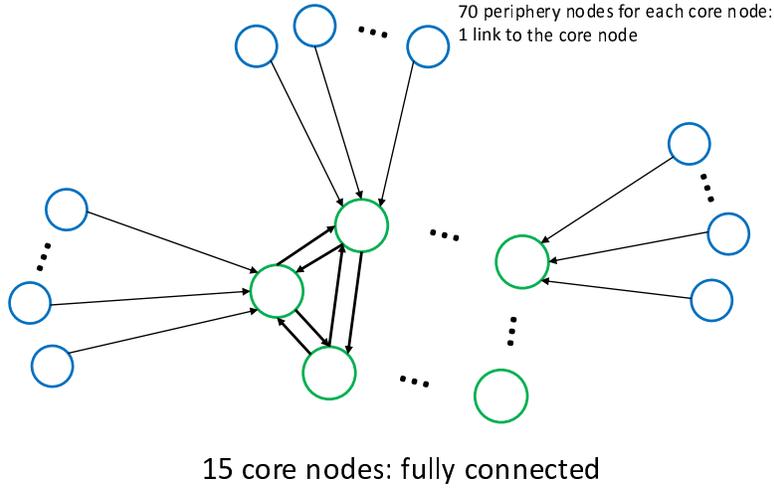}
    \caption{A core-periphery network.} \label{fig:cp-sim}
\end{figure}

\subsection{Numerical Simulations}
To solve MILP~(\ref{eq:MILP2}), we use CVX, a package for specifying and solving 
convex programs~\cite{cvx,GrBo08}. A variety of prior literature, e.g.~\cite{SoBeArGlBe07}, 
suggests that the US interbank network is well modeled as a core-periphery network
that consists of a core of about 15 highly interconnected banks to which most other banks 
connect. Therefore, we test the running time on a core-periphery network
shown in Fig.~\ref{fig:cp-sim}. It contains 15 fully connected core nodes.  
Each core node has 70 periphery nodes. Each periphery node has a single link pointing to 
the corresponding core node. Every node has zero external assets: ${\bf e}=\mathbf{0}$.
All the obligation amounts $L_{i,j}$ are independent uniform random variables.  For each pair of core
nodes $i$ and $j$ the obligation amount $L_{ij}$ is uniformly distributed in $[0,10]$. For a core node $i$
and its periphery node $k$, the obligation amount $L_{ki}$ is uniformly distributed in $[0,1]$.
For a core node $i$, we set the weight $w_i=10$; for a periphery node $k$, we set the weight $w_k=1$.
For this core-periphery network, we generate 100 samples. We run the CVX code on a personal computer
with a 2.66GHz Intel Core2 Duo Processor P8800. The average running time is $1.9s$ and 
the sample standard deviation is $2.0s$. The relative gap between the objective of the solution 
and the optimal objective is less than $10^{-4}$. (This bound is obtained by calculating the optimal
value of the objective for the corresponding linear program, which is an upper bound for the optimal
objective value of the MILP.) We can see that for the core-periphery network, 
MILP~(\ref{eq:MILP2}) can be solved by CVX efficiently and accurately. The
CVX code is given in Appendix~\ref{app:cvx}.

\section{Random Capital Model} \label{sec:random-e}
In the previous sections, we assume that the external asset vector $\bf e$ is
a deterministic vector known by the regulator. However, some applications,
such as stress testing, require forecasting and planning for a wide variety of different contingencies. 
Such applications call for the use of stochastic models for the nodes' asset amounts.
In this case, we aim to solve a stochastic version of Problem~I where $\bf e$ is modeled
as a random vector.  It is assumed that we are able to efficiently obtain independent
samples of this vector.  The remaining parameters---$\bar{\bf p}$, $\Pi$ and $\bf w$---are still
assumed to be deterministic and known, and are defined as in the previous section.
According to Lemma~\ref{lem:uniqueness}, the clearing payment vector that minimizes
the weighted sum of unpaid liabilities is a function of $\bf e$ and $\bf c$, which we denote
as ${\bf p}^\ast({\bf e},{\bf c})$. If $\bf e$ is a random vector, so is
${\bf p}^\ast({\bf e},{\bf c})$. We use $W^\ast({\bf e},{\bf c})$ to denote the corresponding
minimum value of the weighted sum of unpaid liabilities.  If $\bf e$ is a random vector,
then $W^\ast({\bf e},{\bf c})$ is a random variable. Given a total amount of cash $C$, our aim
is to find the optimal cash allocation strategy $\bf c$ to minimize the expectation of
the weighted sum of unpaid liabilities. This can be formulated as the following two-stage stochastic LP:
\begin{align}
& \displaystyle \min_{\bf c} \mathbb{E}_{\bf
e}[W^\ast({\bf e},{\bf c})] \label{eq:SLP}
\\
& \mbox{subject to } \nonumber\\
& \mathbf{1}^T{\bf c} \le C, \nonumber\\
& {\bf c} \geq \mathbf{0}, \nonumber
\end{align}
where
\begin{align}
W^\ast({\bf e},{\bf c})= \mbox{ }& \min_{\bf p} \mathbf {w}^T(\bar{\bf p}-{\bf p}) \label{eq:W}\\
& \mbox{subject to } \nonumber\\
& \mathbf{0}\leq {\bf p} \leq \bar{\bf p}, \nonumber\\
& {\bf p} \leq \Pi^T{\bf p} + {\bf e} + {\bf c}. \nonumber
\end{align}

Even if the joint distribution of ${\bf e}$ is known, the distributions of
${\bf p}^\ast({\bf e},{\bf c})$ and $W^\ast({\bf e},{\bf c})$ cannot be computed in closed form.
In order to solve (\ref{eq:SLP}), we take $M$ independent samples of the asset vector,
denoted as ${\bf e}^1, {\bf e}^2, \cdots, {\bf e}^M$, and use
$\frac{1}{M} \displaystyle \sum_{m=1}^{M} W^\ast({\bf e}^m,{\bf c})$ to approximate
$\mathbb{E}_{\bf e}[W^\ast({\bf e},{\bf c})]$. By the weak law of large numbers, when $M$
is large enough, the sample average is a good estimate of $\mathbb{E}_{\bf e}[W^\ast({\bf e},{\bf c})]$.
This motivates approximating Eq.~(\ref{eq:SLP}) as follows:
\begin{align}
& \displaystyle \min_{\bf c} \frac{1}{M} \displaystyle \sum_{m=1}^{M} W^\ast({\bf e}^m,{\bf c}) \label{eq:Appox-SLP}
\\
& \mbox{subject to } \nonumber\\
& \mathbf{1}^T{\bf c} \le C, \nonumber\\
& {\bf c} \geq \mathbf{0}. \nonumber
\end{align} 
Similar to Theorem~\ref{thm:LP}, the optimization problems~(\ref{eq:W}) and~(\ref{eq:Appox-SLP})
can be combined into one single LP:
\begin{align}
& \displaystyle \max_{{\bf c},{\bf p}^m} \sum_{m=1}^{M}\mathbf{w}^T {\bf p}^m \label{eq:LP4SLP} \\
& \mbox{subject to } \nonumber\\
& \mathbf{1}^T{\bf c} \le C, \nonumber\\
& {\bf c} \geq \mathbf{0},\nonumber\\
& \mathbf{0} \le {\bf p}^m \le \bar{\bf p} \mbox{, for } m=1,2,\cdots,M, \nonumber\\
& {\bf p}^m \le \Pi^T {\bf p}^m + {\bf e}^m + {\bf c} \mbox{, for } m=1,2,\cdots,M. \nonumber
\end{align} 
Since ${\bf c}$ and ${\bf p}^m$ for $m=1,2,\cdots,M$ are all $N$-dimensional vectors,
LP~(\ref{eq:LP4SLP}) contains $MN+N$ variables. The computational complexity of solving an
LP with $MN+N$ variables is $O((MN+N)^3)$~\cite{Ye91}. To achieve a high accuracy, $M$ needs to be
a large number. Then the computational burden is large if we want to solve
LP~(\ref{eq:LP4SLP}) directly. The memory complexity, which is $O(MN^2)$, may also be prohibitive
for large $M$ and $N$. Hence, efficient algorithms to solve LP (\ref{eq:LP4SLP}) are needed.

\subsection{Benders Decomposition} \label{subsec:benders}
If the cash injection
vector $\bf c$ is fixed, then the LP~(\ref{eq:LP4SLP}) can be split into $M$ smaller independent
LPs---one for each sample ${\bf e}^m$---each of which can be solved independently for each ${\bf p}^m$.
In this case, instead of solving an LP with $MN$ variables, we solve $M$ LPs with $N$ variables each,
which significantly reduces the computational complexity. Inspired by this idea, we apply Benders
decomposition to the LP~(\ref{eq:LP4SLP}). Benders decomposition, which is described
in~\cite{Be62,VaWe69,In92}, makes a partition of $\bf c$ and ${\bf p}^m$ ($m=1,2,\cdots,M$)
and allows us to find ${\bf p}^m$ iteratively with fixed $\bf c$ in each step. In fact, for our
problem, Benders decomposition can be further simplified due to some special properties
of~(\ref{eq:LP4SLP}).

From the proof of Lemma~\ref{lem:uniqueness}, we know that for any fixed $\bf c$, the feasible
region of ${\bf p}^m$ is non-empty. Thus, (\ref{eq:LP4SLP}) is equivalent to the following problem:
\begin{align}
& \displaystyle \max_{\bf c} V({\bf c}) \label{eq:bender-main} \\
& \mbox{subject to } \nonumber\\
& \mathbf{1}^T{\bf c} \le C, \nonumber\\
& {\bf c} \geq \mathbf{0}.\nonumber
\end{align}
where
\begin{align}
V({\bf c})= & \displaystyle \max_{{\bf p}^m} \sum_{m=1}^{M}\mathbf{w}^T {\bf p}^m \label{eq:bender-sub} \\
& \mbox{subject to } \nonumber\\
& {\bf p}^m \ge \mathbf{0} \mbox{, for } m=1,2,\cdots,M, \nonumber\\
& {\bf p}^m \le \bar{\bf p} \mbox{, for } m=1,2,\cdots,M, \label{eq:benders-sub_constr1}\\
& {\bf p}^m \le \Pi^T {\bf p}^m + {\bf e}^m + {\bf c} \mbox{, for } m=1,2,\cdots,M.
\label{eq:benders-sub_constr2}
\end{align}
We let ${\boldsymbol \mu}^1,\cdots,{\boldsymbol \mu}^M$ be the dual variables for the $M$
constraints~(\ref{eq:benders-sub_constr1}), and we let 
${\boldsymbol \nu}^1,\cdots,{\boldsymbol \nu}^M$ be the dual variables for the $M$
constraints~(\ref{eq:benders-sub_constr2}).
Then $V({\bf c})$ can be obtained from the following dual problem:
\begin{align}
V({\bf c})= & \displaystyle \min_{{\boldsymbol \mu}^m,{\boldsymbol \nu}^m} \sum_{m=1}^{M}
\left[\bar{\bf p}^T{\boldsymbol \mu}^m + {\bf e}^{mT}{\boldsymbol \nu}^m +
{\bf c}^T{\boldsymbol \nu}^m \right]
\label{eq:bender-sub-dual} \\
& \mbox{subject to } \nonumber\\
& {\boldsymbol \mu}^m \ge \mathbf{0} \mbox{, for } m=1,2,\cdots,M, \nonumber\\
& {\boldsymbol \nu}^m \ge \mathbf{0} \mbox{, for } m=1,2,\cdots,M, \nonumber\\
& {\boldsymbol \nu}^m \ge \Pi {\boldsymbol \nu}^m + \mathbf{w} - {\boldsymbol \mu}^m
 \mbox{, for } m=1,2,\cdots,M. \nonumber
\end{align}
Note that, since $V({\bf c})$ minimizes the objective function of LP~(\ref{eq:bender-sub-dual}) subject to the
constraints of LP~(\ref{eq:bender-sub-dual}), we have that $V({\bf c})$ is the greatest lower bound of this
objective function, subject to these constraints. Therefore, LP~(\ref{eq:bender-main})
is equivalently rewritten as the maximization of the lower bound to the objective function of LP~(\ref{eq:bender-sub-dual}),
subject to the constraints of both LP~(\ref{eq:bender-main}) and LP~(\ref{eq:bender-sub-dual}):
\begin{align}
& \displaystyle \max_{{\bf c},\theta} \theta \label{eq:bender-master} \\
& \mbox{subject to } \nonumber\\
& \mathbf{1}^T{\bf c} \le C, \nonumber \\
& {\bf c} \geq \mathbf{0}, \nonumber \\
& \theta \le \sum_{m=1}^{M} \left[ \bar{\bf p}^T{\boldsymbol \mu}^m + {\bf e}^{mT}{\boldsymbol \nu}^m +
{\bf c}^T{\boldsymbol \nu}^m\right]
\label{eq:bender-master-contraint1} \\
& \mbox{for all } ({\boldsymbol \mu}^m,{\boldsymbol \nu}^m)
\mbox{ in the feasible region of (\ref{eq:bender-sub-dual}). } \nonumber
\end{align}
LP~(\ref{eq:bender-master}), the equivalent version of~(\ref{eq:LP4SLP}), has an infinite number of constraints
in the form of~(\ref{eq:bender-master-contraint1}), because constraint~(\ref{eq:bender-master-contraint1}) must
be satisfied by every pair $({\boldsymbol \mu}^m,{\boldsymbol \nu}^m)$ from the feasible region of
LP~(\ref{eq:bender-sub-dual}). The key idea is solving a relaxed version of~(\ref{eq:bender-master})
by ignoring all but a few of the constraints~(\ref{eq:bender-master-contraint1}).
Assume the optimal solution of this relaxed program is $({\bf c}^\ast,\theta^\ast)$.
If the solution satisfies all the ignored constraints, the optimal solution has been found;
otherwise, we generate a new constraint by solving (\ref{eq:bender-sub-dual}) with fixed
${\bf c}={\bf c}^\ast$ and add it to the relaxed problem.
Here is the summary of the Benders decomposition algorithm:
\begin{enumerate}
\item 
Initialization: set $\theta^0 \leftarrow -\infty$, $K\leftarrow0$, ${\bf c}^0 \leftarrow \mathbf{0}$, $l\leftarrow0$.

\item
Fix ${\bf c}^l$, solve the following $M$ sub-programs for $m=1,2,\cdots,M$:
\begin{align}
V^m= & \displaystyle \min_{{\boldsymbol \mu}^m,{\boldsymbol \nu}^m} \sum_{m=1}^{M}
\left[\bar{\bf p}^T{\boldsymbol \mu}^m + {\bf e}^{mT}{\boldsymbol \nu}^m +
{\bf c}^{lT}{\boldsymbol \nu}^m \right]
\nonumber \\
& \mbox{subject to } \nonumber\\
& {\boldsymbol \mu}^m \ge \mathbf{0} \mbox{, for } m=1,2,\cdots,M, \nonumber\\
& {\boldsymbol \nu}^m \ge \mathbf{0} \mbox{, for } m=1,2,\cdots,M, \nonumber\\
& {\boldsymbol \nu}^m \ge \Pi {\boldsymbol \nu}^m + \mathbf{w} - {\boldsymbol \mu}^m
 \mbox{, for } m=1,2,\cdots,M. \nonumber
\end{align}
Denote the solution as $({\boldsymbol \mu}^{m\ast},{\boldsymbol \nu}^{m\ast})$, for $m=1,2,\cdots,M$.

\item
If $\displaystyle \sum_{m=1}^{M}V^m = \theta^l$, terminate and ${\bf c}^l$ is the optimal.

\item
Set $l\leftarrow l+1$, $K\leftarrow K+1$, $({\boldsymbol \mu}^m(K),{\boldsymbol \nu}^m(K)) \leftarrow
 ({\boldsymbol \mu}^{m\ast},{\boldsymbol \nu}^{m\ast})$, 
for $m=1,2,\cdots,M$.

\item
Solve the following master problem:
\begin{align}
& \displaystyle \max_{{\bf c},\theta} \theta \label{eq:bender-master2} \\
& \mbox{subject to } \nonumber\\
& \mathbf{1}^T{\bf c} \le C, \nonumber \\
& {\bf c} \geq \mathbf{0}, \nonumber \\
& \theta \le \sum_{m=1}^{M} \left[ \bar{\bf p}^T{\boldsymbol \mu}(k)^m 
+ {\bf e}^{mT}{\boldsymbol \nu}(k)^m +{\bf c}^T{\boldsymbol \nu}(k)^m\right]
\nonumber \\
& \mbox{for } k=1,2,\cdots,K. \nonumber
\end{align} 
Denote the solution as $({\bf c}^\ast, \theta^\ast)$. Set $\theta^l \leftarrow \theta^\ast$, ${\bf c}^l \leftarrow {\bf c}^\ast$. Then go to Step 2.
\end{enumerate}

In this algorithm, at each iteration, we solve $M+1$ LPs with $N$ variables instead of one LP with $MN+N$ variables,
which saves both computational complexity and memory cost. Note that, comparing to the general form of Benders
decomposition (Section 2.3 in~\cite{VaWe69}), the above algorithm is simpler. Since LP~(\ref{eq:bender-sub}) is
always feasible and bounded, it is not necessary to consider constraint~(22b) in~\cite{VaWe69}. Step~(2') in~\cite{VaWe69}
can also be removed since~(\ref{eq:bender-master2}) is always bounded.

\subsection{Projected Stochastic Gradient Descent} \label{subsec:stochastic-gd}

In this section, we introduce the projected stochastic gradient descent method to solve~(\ref{eq:SLP}).
This is an online algorithm, which allows us to handle one sample at a time, without building a huge linear
program. The basic idea is that for each sample ${\bf e}^m$, we move the solution $\bf c$ along
the direction of the negative gradient of $W^\ast({\bf e}^m,{\bf c})$ with respect to $\bf c$ and
then project the result onto the set defined by the constraints of~(\ref{eq:SLP}). This procedure will converge to
the optimal solution if the step size is selected properly~\cite{Bo98}. 
The algorithm proceeds as follows. At iteration $m$,
\begin{enumerate}
\item
Sample an asset vector ${\bf e}^m$.
\item
Move $\bf c$ along the negative gradient of $W^\ast({\bf e}^m,{\bf c}^{m-1})$ 
according to the following equation:
\begin{equation} \label{eq:sgd}
\tilde{\bf c}^m = {\bf c}^{m-1}-\gamma^m \nabla_{\bf c} W^\ast({\bf e}^m,{\bf c}^{m-1}).
\end{equation}
\item
Set ${\bf c}^{m}$ as the projection of $\tilde{\bf c}^m$ onto the set 
$\{{\bf c}:\,\,\,\mathbf{1}^T{\bf c}=C,\,\,\, {\bf c} \geq \mathbf{0}\}$.
\end{enumerate}

According to \cite{Bo98}, step size $\gamma^m$ should satisfy the condition that 
$\displaystyle \sum_{m=1}^{\infty} (\gamma^m)^2 < \infty$ and 
$\displaystyle \sum_{m=1}^{\infty} \gamma^m = \infty$. 
Thus, a proper choice could be $\gamma^m=1/m$.

Note that $W^\ast({\bf e}^m,{\bf c}^{m-1}) = {\bf w}^T\bar{\bf p} - U({\bf e}^m,{\bf c}^{m-1})$, where
\begin{align}
U({\bf e}^m,{\bf c}^{m-1}) = \mbox{ }& \max_{\bf p} \mathbf {w}^T{\bf p} \label{eq:U}\\
& \mbox{subject to } \nonumber\\
& \mathbf{0}\leq {\bf p} \leq \bar{\bf p}, \nonumber\\
& {\bf p} \leq \Pi^T{\bf p} + {\bf e}^m + {\bf c}^{m-1}. \nonumber
\end{align}

To obtain the gradient of $W^\ast({\bf e}^m,{\bf c}^{m-1})$ in Step 2, 
we consider the dual problem of LP (\ref{eq:U}):

\begin{align}
U({\bf e}^m,{\bf c}^{m-1}) = & \displaystyle \min_{{\boldsymbol \mu},{\boldsymbol \nu}} 
\left[\bar{\bf p}^T{\boldsymbol \mu} + {\bf e}^{mT}{\boldsymbol \nu} +
{\bf c}^{(m-1)T}{\boldsymbol \nu} \right]
\label{eq:gd-dual} \\
& \mbox{subject to } \nonumber\\
& {\boldsymbol \mu} \ge \mathbf{0}, \nonumber\\
& {\boldsymbol \nu} \ge \mathbf{0}, \nonumber\\
& {\boldsymbol \nu} \ge \Pi {\boldsymbol \nu} + \mathbf{w} - {\boldsymbol \mu}. \nonumber
\end{align}

Assuming that $({\boldsymbol \mu}^\ast,{\boldsymbol \nu}^\ast)$ is a solution of~(\ref{eq:gd-dual}), 
we have $\nabla_{\bf c} W^\ast({\bf e}^m,{\bf c}^{m-1}) = 
-\nabla_{\bf c} U({\bf e}^m,{\bf c}^{m-1}) = -{\boldsymbol \nu}^\ast$.

In Step 3, we find the projection of $\tilde{\bf c}^m$ using the following quadratic program:

\begin{align}
{\bf c}^m = & \displaystyle \arg\min_{\bf c} \Vert{\bf c}-\tilde{\bf c}^m\Vert_2^2 \label{eq:quadprog} \\
& \mbox{subject to } \nonumber\\
& \mathbf{1}^T{\bf c} = C, \nonumber\\
& {\bf c} \geq \mathbf{0}.\nonumber
\end{align}

A fast algorithm for quadratic program (\ref{eq:quadprog}) is given in Appendix~\ref{app:quadprog}.

Thus, at each iteration in this projected stochastic gradient descent method, instead of solving
LP~(\ref{eq:LP4SLP}) 
which contains $MN+N$ variables, we solve one $N$-variable LP and one $N$-variable quadratic program. 
This algorithm is memory efficient because it requires no storage except the current solution of $\bf c$.

\section{Distributed Algorithms for Problem I with Proportional Payment Mechanism}
\label{sec:DA}

We showed in Theorem~\ref{thm:LP} that Problem~I without bankruptcy costs is equivalent to a linear program, and therefore can be solved exactly, for any network topology, using standard LP solvers. In some scenarios, however, this approach may be impractical or undesirable, as it requires the solver to know the entire network structure, namely, the net external assets of every institution, as well as the amounts owed by each institution to each other institution. We now adapt our framework to applications where it is necessary to avoid centralized data gathering and computation. We propose a distributed algorithm to solve our linear program. The algorithm is iterative and is based on message passing between each node and its neighbors. During each iteration of the algorithm, each node only needs to receive a small amount of data from its neighbors, perform simple calculations, and transmit a small amount of data to its neighbors.
During the message passing, no node will reveal to any other node any proprietary information on its asset values, the amounts owed to other nodes, or the amounts owed by other nodes.

Our algorithm can be used both to monitor financial networks and to simulate stress-testing scenarios. The integrity of the process can be enforced by the supervisory authorities through auditing.

While the algorithm is slower than standard centralized LP solvers, simulations suggest its practicality for the US banking system which we model as a core-periphery network with 15 core nodes and 1050 periphery nodes.

\subsection{Problem~I}
\label{subsec:DALP1}

\subsubsection{A Distributed Algorithm}

To develop a distributed algorithm for LP~(\ref{eq:LP}-\ref{eq:LP-constraint2}), we formulate its dual problem and solve it via gradient descent because the dual problem has simpler constraints which are easily decomposable. It turns out that every iteration of the gradient descent involves only local computations, which enables a distributed implementation. 

In order to apply the gradient descent method to the dual problem, we need the objective function in~(\ref{eq:LP}) to be strictly concave, which would guarantee that the dual problem is differentiable at any point~\cite{Fl95}. However, the objective function of LP~(\ref{eq:LP}) is not strictly concave and so we apply the Proximal Optimization Algorithm~\cite{LiNe06,BeTs89}. The basic idea is to add quadratic terms to the objective function. The quadratic terms will converge to zero so that we make the objective function strictly concave without changing the optimal solution. 

We introduce two $N \times 1$ vectors $\bf y$ and $\bf z$ and add two quadratic terms 
\begin{equation}
\Vert {\bf p}-{\bf y}\Vert^2 = \sum_{i=1}^{N} (p_i-y_i)^2 \mbox{, } \qquad
\Vert {\bf c}-{\bf z}\Vert^2 = \sum_{i=1}^{N} (c_i-z_i)^2 \nonumber
\end{equation}
to (\ref{eq:LP}). Then we proceed as follows.\\
{\bf Algorithm} $\mathcal{P}$:\\
At the $t$-th iteration,
\begin{itemize}
\item
{\bf P1)} Fix ${\bf y}={\bf y}(t)$ and ${\bf z}={\bf z}(t)$ and maximize the objective function with respect to $\bf p$ and $\bf c$:
\begin{align}
& \displaystyle \max_{{\bf p}, {\bf c}} {\bf w}^T{\bf p} - \Vert {\bf p}-{\bf y}\Vert^2 - \Vert {\bf c}-{\bf z}\Vert^2 \label{eq:LPQ}\\
& \mbox{subject to } \nonumber\\
& \mathbf{1}^T{\bf c} \le C, \label{eq:LPQ-constraint1}\\
& {\bf c} \geq \mathbf{0}, \nonumber\\
& \mathbf{0}\leq {\bf p} \leq \bar{\bf p}, \nonumber\\
& {\bf p} \leq \Pi^T{\bf p} + {\bf e} + {\bf c}. \label{eq:LPQ-constraint2}
\end{align}

Note that since the objective function is strictly concave, a unique solution exists. Denote it as~${\textbf {p}^\ast}$ and~${\textbf {c}^\ast}$.

\item
{\bf P2)} Set ${\bf y}(t+1)={\textbf {p}^\ast}$, ${\bf z}(t+1)={\textbf {c}^\ast}$.
\end{itemize}

It is proved in Proposition 4.1 in \cite{BeTs89} that algorithm $\mathcal{P}$ will converge to the optimal solution of LP (\ref{eq:LP}-\ref{eq:LP-constraint2}).

\subsubsection{Implementation of Algorithm $\cal{P}$}

In Step {\bf P1}, for fixed $\bf y$ and $\bf z$, the objective function of~(\ref{eq:LPQ}) is strictly concave so that the dual problem is differentiable at any point~\cite{Fl95}. Hence, we can solve the dual problem using gradient descent, as follows.

Let a scalar $\lambda$ and an $N \times 1$ vector ${\bf q}$ be Lagrange multipliers for constraints~(\ref{eq:LPQ-constraint1}) and~(\ref{eq:LPQ-constraint2}), respectively. We define the Lagrangian as follows:
\begin{equation}
L({\bf p},{\bf c},\lambda,{\bf q},{\bf y},{\bf z}) = 
{\bf w}^T{\bf p}-\lambda(\mathbf{1}^T{\bf c}-C)-{\bf q}^T((I_{N\times N}-\Pi^T){\bf p}-{\bf e}-{\bf c})-\Vert{\bf p}-{\bf y}\Vert^2-\Vert{\bf c}-{\bf z}\Vert^2
\label{eq:Lagrangian}
\end{equation}
where $\lambda$ and $\bf q$ are non-negative and $I_{N\times N}$ is an $N\times N$ identity matrix.
We further expand~(\ref{eq:Lagrangian}):
\begin{align}
L({\bf p},{\bf c},\lambda,{\bf q},{\bf y},{\bf z}) \nonumber&=
\displaystyle
\sum_{i=1}^{N}w_ip_i-\lambda\sum_{i=1}^{N}c_i+\lambda C-\sum_{i=1}^{N}q_i(p_i-\sum_{j=1}^{N}\Pi_{ji}p_j-e_i-c_i)\nonumber
\\&\qquad -\sum_{i=1}^N{(p_i-y_i)^2}-\sum_{i=1}^N{(c_i-z_i)^2}\label{eq:lagragian1}\\
&= - \sum_{i=1}^{N}\left[p_i^2-(w_i-q_i+2y_i+\sum_{j=1}^{N}q_j\Pi_{ij})p_i \right]\nonumber
\\&\qquad - \sum_{i=1}^{N} \left[c_i^2-(q_i-\lambda+2z_i)c_i \right] + \sum_{i=1}^{N} \left[q_ie_i-y_i^2-z_i^2 \right] + \lambda C  \label{eq:lagragian2}
\end{align}
To obtain Eq. (\ref{eq:lagragian2}) from Eq. (\ref{eq:lagragian1}),
we use the following equation: $\displaystyle \sum_{i=1}^N\sum_{j=1}^N {q_i\Pi_{ji}p_j}=\sum_{i=1}^N\sum_{j=1}^N {q_j\Pi_{ij}p_i}$.
Then the objective function of the dual problem is:
\begin{equation}
D(\lambda,\mathbf{q},\mathbf{y},\mathbf{z})=\displaystyle
\max_{\mathbf{0}\le\mathbf{p}\le\mathbf{\bar{p}},\mathbf{c}\ge\mathbf{0}}
L(\mathbf{p},\mathbf{c},\lambda,\mathbf{q},\mathbf{y},\mathbf{z})
\label{eq:dualobj}
\end{equation}

In Eq.~(\ref{eq:lagragian2}), the term $q_j\Pi_{ij}$ is 0 if node $i$ is not a borrower of $j$. Thus, if node $i$ receives all the $q_j$ from its lenders, it can determine $p_i$ and $c_i$ to achieve the maximum of the Lagrangian $L(\mathbf{p},\mathbf{c},\lambda,\mathbf{q},\mathbf{y},\mathbf{z})$.

Given $\bf y$ and $\bf z$, the dual problem of~(\ref{eq:LPQ}) is then minimizing~(\ref{eq:dualobj}) over Lagrange multipliers $\lambda$ and $\bf q$:
\begin{equation}
\min_{\lambda \ge 0, \bf{q} \ge \mathbf{0}}D(\lambda,\mathbf{q},\mathbf{y},\mathbf{z}) \label{eq:dual}
\end{equation}

The dual problem is differentiable at any point since the objective function of the primal is strictly concave~\cite{Fl95}. Hence, gradient descent iterations can be applied to solve the dual.

At iteration $u$, $\lambda=\lambda(u)$ and ${\bf q} = {\bf q}(u)$. Then the gradients of $D$ with respect to $\lambda$ and ${\bf q}$ at this point are:
\begin{equation}
\frac{\partial D}{\partial \lambda} = C - \mathbf{1}^T{\bf c}(u), \nonumber
\end{equation}

\begin{equation}
\frac{\partial D}{\partial {\bf q}} = {\bf e} + {\bf c}(u) - (I_{N\times N}-\Pi^T){\bf p}(u), \nonumber
\end{equation}
where ${\bf p}(u) $ and ${\bf c}(u)$ solve (\ref{eq:dualobj}) for $\lambda=\lambda(u)$ and ${\bf q} = {\bf q}(u)$:
\begin{equation}
\label{eq:update-p}
 p_i(u) = \left\{
\begin{array}{cl}
0 & \text{if } \tilde{p}_i(u) < 0\\
\bar{p}_i & \text{if } \tilde{p}_i(u)>\bar{p}_i\\
\tilde{p}_i(u) & \text{otherwise}
\end{array} \right.
\end{equation}
where
$\tilde{p}_i(u)=y_i(t)+\frac{1}{2}(w_i-q_i(u)+\displaystyle \sum_{j \in \mathcal{C}_i} q_j(u)\Pi_{ij})$, and
\begin{equation}
\label{eq:update-c} c_i(u)=\left[z_i(t)+\frac{1}{2}(q_i(u)-\lambda(u)) \right]^+.
\end{equation}
Therefore, taking into account the non-negativity of $\lambda$ and $\bf q$, the gradient descent equations are:
\begin{equation}
\lambda(u+1)=\left[\lambda(u)-\alpha(C - \sum_{i=1}^{N}c_i(u)) \right]^+, \label{eq:update-lambda}
\end{equation}
\begin{equation}
q_i(u+1)=\left[q_i(u)-\beta(e_i+c_i(u)-p_i(u)+\sum_{j=1}^{N}\Pi_{ji}p_j(u)) \right]^+, \mbox{ for } i=1,2,\cdots,N \label{eq:update-q}
\end{equation}
where $\alpha$ and $\beta$ are the step sizes, and $[x]^+=\max\{0,x\}$. For fixed $\bf y$ and $\bf z$, the dual update will converge to the minimizer of $D$ as $u \rightarrow \infty$, if the step size is small enough~\cite{BeTs89}.

From~(\ref{eq:update-lambda}), we notice that in order to update $\lambda$, $c_i$ is required from all the $N$ nodes. It means at each iteration $t$, each node should send $c_i(u)$ to a central node so that the central node could update $\lambda$ and send it back to every node in the system.

If node $j$ is not a borrower of node $i$, then $\Pi_{ji}p_j(u)=0$; otherwise, $\Pi_{ji}p_j(u)$ represents
the amount of money that node $j$ pays to node $i$ at $u$-th iteration. Hence, with the information of $\Pi_{ji}p_j(u)$ from all its borrowers, node $i$ is able to update $q_i$ based on~(\ref{eq:update-q}).

\begin{figure}
    \centering
    \includegraphics[width=\textwidth]{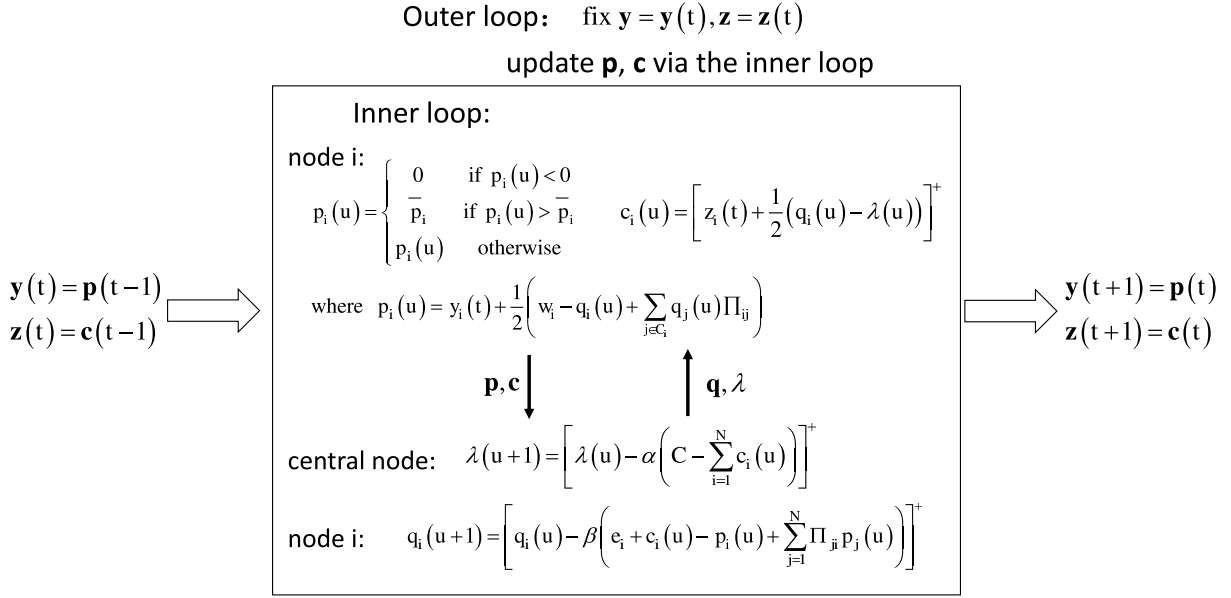}
    \caption{Duality-based approach.}
\label{fig:dual}
\end{figure}


\subsubsection{A More Efficient Algorithm}
As shown in Fig.~\ref{fig:dual}, in the algorithm, we first fix $\bf y$ and $\bf z$ and solve~(\ref{eq:LPQ})
by updating $\lambda$, $\bf q$ and $\bf p$, $\bf c$ iteratively until they converge.
Then we update $\bf y$ and $\bf z$. This is a two-stage iteration, which is likely to slow down
the convergence of the entire algorithm as too many dual updates are wasted for each
fixed $\bf y$ and $\bf z$~\cite{LiNe06}. 
To avoid the two-stage iteration structure, we consider the following algorithm.\\
{\bf Algorithm} $\mathcal{A}$:\\
At the $t$-th iteration,
\begin{itemize}
\item
{\bf A1)} Fix ${\bf y}={\bf y}(t)$ and ${\bf z}={\bf z}(t)$,
maximize $L$ with respect to $\textbf{p}$ and $\textbf{c}$,
\begin{equation}
\left[\textbf{p}(t),\textbf{c}(t) \right]=\displaystyle
\arg \max_{\mathbf{0}\le\mathbf{p}\le\mathbf{\bar{p}},\mathbf{c}\ge\mathbf{0}}
L(\textbf{p},\textbf{c},\lambda(t), \textbf{q}(t),\textbf{y}(t),\textbf{z}(t)).
\nonumber
\end{equation}
\item
{\bf A2)} Update Lagrange multipliers $\lambda(t+1)$ and $\textbf{q}(t+1)$ by 
\begin{equation}
\lambda(t+1)=\left[\lambda(t)-\alpha\left(C - \sum_{i=1}^{N}c_i(t)\right) \right]^+, \label{eq:update-lambda-A}
\end{equation}
\begin{equation}
q_i(t+1)=\left[q_i(t)-\beta(e_i+c_i(t)-p_i(t)+\sum_{j=1}^{N}\Pi_{ji}p_j(t)) \right]^+, \mbox{ for } i=1,2,\cdots,N \label{eq:update-q-A}
\end{equation}
\item
{\bf A3)} Update $\textbf{y}$ and $\textbf{z}$ with
\begin{equation}
\left[{\textbf y}(t+1), {\textbf z}(t+1) \right] =\displaystyle
\arg \max_{\mathbf{0}\le\mathbf{p}\le\mathbf{\bar{p}},\mathbf{c}\ge\mathbf{0}}
L(\textbf{p},\textbf{c},\lambda(t+1), \textbf{q}(t+1),\textbf{y}(t),\textbf{z}(t)).\nonumber
\end{equation}
\end{itemize}
In algorithm $\mathcal{A}$, instead of an infinite number of dual updates, we only update
Lagrange multipliers $\lambda$ and $\bf q$ once for each fixed $\bf y$ and $\bf z$.
The following theorem guarantees the convergence of algorithm $\mathcal{A}$.

\begin{thm} \label{thm:disalg}
Algorithm $\mathcal{A}$ will converge to the optimal solution of LP (\ref{eq:LP}-\ref{eq:LP-constraint2})
provided the step sizes $\alpha$ and $\beta$ are sufficiently small.
\end{thm}
Theorem~\ref{thm:disalg} is an extension of Proposition~4 in~\cite{LiNe04}. 

\begin{figure}
    \centering
    \includegraphics[width=0.45\textwidth]{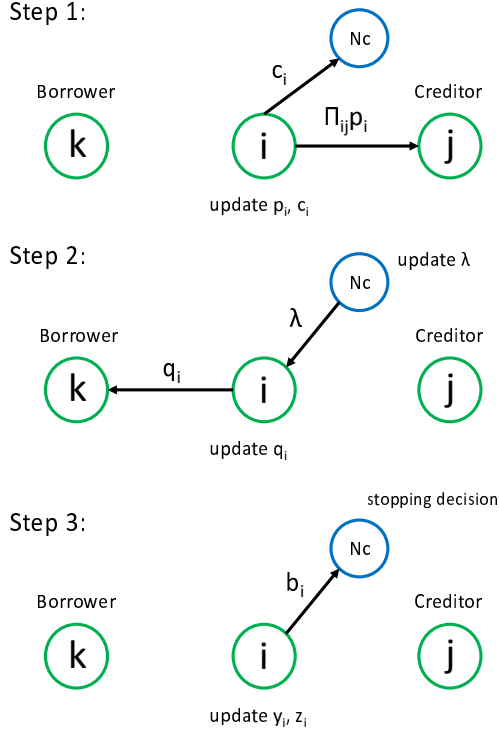}
    \caption{The $t$-th iteration of the distributed algorithm $\cal{A}$ for a fixed maximum total amount of injected cash.}
\label{fig:algo-bdd}
\end{figure}

\subsubsection{Implementation of Algorithm $\cal{A}$} \label{subsubsec:Prob1A}
Assume $\mathcal{B}_i$ and $\mathcal{C}_i$ are the sets of borrowers and creditors of node $i$
respectively. Then the $t$-th iteration of algorithm $\mathcal{A}$ is as follows.\\
\begin{enumerate}
\item
For each node $i$, fix $y_i=y_i(t)$, $z_i=z_i(t)$, $\lambda=\lambda(t)$ and ${\bf q}={\bf q}(t)$, and calculate $p_i$
and $c_i$:
\begin{equation}
\nonumber
 p_i(t) = \left\{
\begin{array}{cl}
0 & \text{if } \tilde{p}_i(t) < 0\\
\bar{p}_i & \text{if } \tilde{p}_i(t)>\bar{p}_i\\
\tilde{p}_i(t) & \text{otherwise,}
\end{array} \right.
\end{equation}
where
$\tilde{p}_i(t)=y_i(t)+\frac{1}{2}\left(w_i-q_i(t)+\displaystyle \sum_{j \in \mathcal{C}_i} q_j(t)\Pi_{ij}\right)$, and
\begin{equation}
\nonumber c_i(t)=\left[z_i(t)+\frac{1}{2}(q_i(t)-\lambda(t)) \right]^+.
\end{equation}
Then send $\Pi_{ij}p_i(t)$ to every node $j\in\mathcal{C}_i$, and send the updated $c_i(t)$ to node $N_c$.
\item
Each node $i$ receives $\Pi_{ki}p_k(t)$ from every $k\in\mathcal{B}_i$ and updates $q_i$:
\begin{equation}\nonumber
q_i(t+1)=\left[q_i(t)+\beta(p_i(t)-e_i-c_i(t)-\displaystyle \sum_{k\in \mathcal{B}_i}\Pi_{ki}p_k(t)) \right]^+.\nonumber
\end{equation}
Then each node $i$ sends the updated $q_i(t+1)$ to every node $k\in\mathcal{B}_i$. 

Node $N_c$ receives $c_i$ from all nodes $i$ and updates $\lambda$:
\begin{equation}
\nonumber \lambda(t+1)=\left[\displaystyle \lambda(t) + \alpha\left(\sum_{i=1}^N c_i(t)-C\right) \right]^+.
\end{equation}
Then node $N_c$ send the updated $\lambda(t+1)$ to every node $i$.

\item
Every node $i$ receives $q_j(t+1)$ from each $j\in\mathcal{C}_i$ and receives $\lambda(t+1)$ from node $N_c$,
then updates $y_i$ and $z_i$:
\begin{equation}
\nonumber
 y_i(t+1) = \left\{
\begin{array}{cl}
0 & \text{if } \tilde{y}_i(t+1) < 0\\
\bar{p}_i & \text{if } \tilde{y}_i(t+1)>\bar{p}_i\\
\tilde{y}_i(t+1) & \text{o.w. }
\end{array} \right.
\end{equation}
where $\tilde{y}_i(t+1)=y_i(t)+\frac{1}{2}(w_i-q_i(t+1)+\displaystyle \sum_{j \in \mathcal{C}_i} q_j(t+1)\Pi_{ij})$, and
\begin{equation}
\nonumber z_i(t+1)=\left[\tilde{z}_i(t+1) \right]^+,
\end{equation}
where $\tilde{z}_i(t+1)=z_i(t)+\frac{1}{2}(q_i(t+1)-\lambda(t+1))$.

Every node $i$ then checks the conditions
$|\tilde{y}_i(t+1)-\tilde{y}_i(t)|<\delta_1$ and $|\tilde{z}_i(t+1)-\tilde{z}_i(t)|<\delta_2$. If both conditions hold,
node $i$ sets $b_i=1$; otherwise it sets $b_i=0$. It then sends $b_i$ to the central node $N_c$.
If $b_i=1$ for all $i$, then $N_c$ directs all nodes to terminate the algorithm.
\end{enumerate}
These steps are illustrated in Fig.~\ref{fig:algo-bdd}.

In Step 3, $\delta_1$ and $\delta_2$ are the stopping tolerances, which are usually set as small positive
numbers according to the accuracy requirement. We utilize $\tilde{\bf y}$ and $\tilde{\bf z}$ rather than
their projections ${\bf y}$ and ${\bf z}$ in the stopping criterion because the convergence of
$\tilde{\bf y}$ and $\tilde{\bf z}$ implies the convergence of the Lagrange multipliers
${\bf q}$ and $\lambda$, whereas the convergence of ${\bf y}$ and ${\bf z}$ does not.

In the implementation of algorithm $\cal{A}$, we include a central node. At each iteration
 the central node has two functions. One is to sum the $c_i(t)$ and calculate $\lambda(t+1)$ in Step 2;
the other is to test whether $b_i=1$ for all nodes $i$ in Step 3. For both functions, the central
node only collects a small amount of data and performs simple calculations. We could entirely
exclude the central node by calculating the sum of $c_i(t)$ and communicating the stopping sign
in a distributed way, at the cost of added computational burden during each iteration.

\subsection{Lagrange Formulation of Problem~I} \label{subsec:L4Prob1}
We now apply the duality-based distributed algorithm to LP~(\ref{eq:LP2}), the Lagrange formulation
of Problem~I. Note that now $\lambda$ represents the importance of the injected cash amount in the overall
cost function.  The algorithm is similar to Section~\ref{subsec:DALP1} except for the fact that
$\lambda$ is not updated at each iteration because $\lambda$ is fixed and given.
Similar to~(\ref{eq:Lagrangian}), we define the Lagrangian as:
\begin{align}
L({\bf p},{\bf c},{\bf q},{\bf y},{\bf z}) &=
{\bf w}^T{\bf p}-\lambda\mathbf{1}^T{\bf c}-{\bf q}^T((I_{N\times N}-\Pi^T){\bf p}-{\bf e}-{\bf c})-\Vert{\bf p}-{\bf y}\Vert^2-\Vert{\bf c}-{\bf z}\Vert^2 \nonumber\\&=
\displaystyle
- \sum_{i=1}^{N}\left[p_i^2-(w_i-q_i+2y_i+\sum_{j=1}^{N}q_j\Pi_{ij})p_i \right]\nonumber
\\&\qquad - \sum_{i=1}^{N} \left[c_i^2-(q_i-\lambda+2z_i)c_i\right] + \sum_{i=1}^{N} \left[q_ie_i-y_i^2-z_i^2 \right]  \label{eq:Lagrangian2}
\end{align}
The objective function of dual problem is:
\begin{equation}
D(\mathbf{q},\mathbf{y},\mathbf{z})=\displaystyle
\max_{\mathbf{0}\le\mathbf{p}\le\mathbf{\bar{p}},\mathbf{c}\ge\mathbf{0}}
L(\mathbf{p},\mathbf{c},\mathbf{q},\mathbf{y},\mathbf{z})
\nonumber
\end{equation}
Then the dual problem is:
\begin{equation}
\min_{\bf{q} \ge \mathbf{0}}D(\mathbf{q},\mathbf{y},\mathbf{z}) \nonumber
\end{equation}
The Lagrange multipliers $\bf q$ are updated by~(\ref{eq:update-q-A}),
where $\bf p$ and $\bf c$ maximize Lagrangian~(\ref{eq:Lagrangian2}).

\begin{figure}
    \centering
    \includegraphics[width=0.45\textwidth]{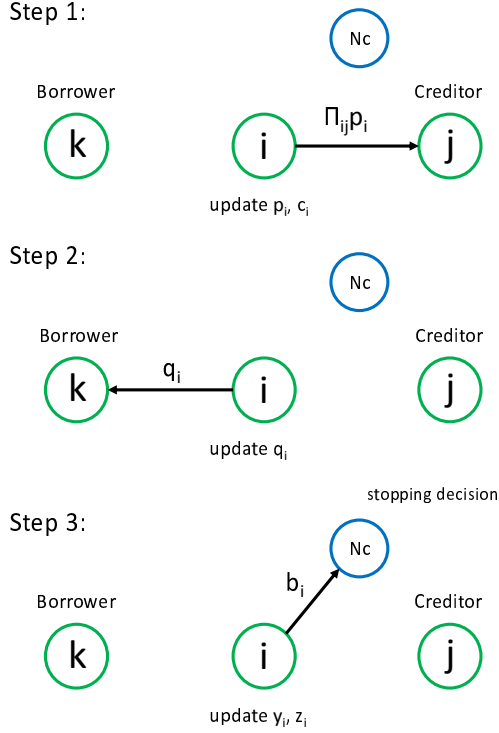}
    \caption{The $t$-th iteration of the distributed algorithm $\cal{A}'$ that includes optimizing the total amount
    of injected cash.} \label{fig:algo-ubdd}
\end{figure}

\subsubsection{Implementation of Algorithm $\cal{A}'$} \label{subsubsec:L4Prob1A}
Our algorithm for this problem is a simple modification of Algorithm $\cal{A}$.
We call it Algorithm $\cal{A}'$.  Its $t$-th iteration is as follows.
\begin{enumerate}
\item
Each node $i$ fixes $y_i=y_i(t)$, $z_i=z_i(t)$, and ${\bf q}={\bf q}(t)$, and calculates $p_i$
and $c_i$:
\begin{equation}
\nonumber
 p_i(t) = \left\{
\begin{array}{cl}
0 & \text{if } \tilde{p}_i(t) < 0\\
\bar{p}_i & \text{if } \tilde{p}_i(t)>\bar{p}_i\\
\tilde{p}_i(t) & \text{o.w. }
\end{array} \right.
\end{equation}
where
$\tilde{p}_i(t)=y_i(t)+\frac{1}{2}(w_i-q_i(t)+\displaystyle \sum_{j \in \mathcal{C}_i} q_j(t)\Pi_{ij})$, and
\begin{equation}
\nonumber c_i(t)=\left[z_i(t)+\frac{1}{2}(q_i(t)-\lambda) \right]^+.
\end{equation}
Then each node $i$ sends $\Pi_{ij}p_i(t)$ to every node $j\in\mathcal{C}_i$.

\item
Each node $i$
receives $\Pi_{ki}p_k(t)$ from every $k\in\mathcal{B}_i$ and updates $q_i$:
\begin{equation}\nonumber
q_i(t+1)=\left[q_i(t)+\beta(p_i(t)-e_i-c_i(t)-\displaystyle \sum_{k\in \mathcal{B}_i}\Pi_{ki}p_k(t)) \right]^+.\nonumber
\end{equation}
Then each node $i$ sends the updated $q_i(t+1)$ to every $k\in\mathcal{B}_i$.

\item
Each node $i$
receives $q_j(t+1)$ from every $j\in\mathcal{C}_i$ and updates
$y_i$ and $z_i$:
\begin{equation}
\nonumber
 y_i(t+1) = \left\{
\begin{array}{cl}
0 & \text{if } \tilde{y}_i(t+1) < 0\\
\bar{p}_i & \text{if } \tilde{y}_i(t+1)>\bar{p}_i\\
\tilde{y}_i(t+1) & \text{o.w. }
\end{array} \right.
\end{equation}
where $\tilde{y}_i(t+1)=y_i(t)+\frac{1}{2}(w_i-q_i(t+1)+\displaystyle \sum_{j \in \mathcal{C}_i} q_j(t+1)\Pi_{ij})$, and
\begin{equation}
\nonumber z_i(t+1)=\left[\tilde{z}_i(t+1)\right]^+,
\end{equation}
where $\tilde{z}_i(t+1)=z_i(t)+\frac{1}{2}(q_i(t+1)-\lambda)$.

Each node $i$ checks the conditions $|\tilde{y}_i(t+1)-\tilde{y}_i(t)|<\delta_1$ and
$|\tilde{z}_i(t+1)-\tilde{z}_i(t)|<\delta_2$.  If both conditions hold, it sets $b_i=1$;
otherwise, it sets $b_i=0$. It then sends $b_i$ to the central node $N_c$.
If $b_i=1$ for all $i$ then $N_c$ asks all nodes to terminate the algorithm.
\end{enumerate}
These steps are illustrated in Fig.~\ref{fig:algo-ubdd}.

\subsection{Numerical Results}
\subsubsection{Example 1: A Four-Node Network}
In this section, we illustrate the convergence of our distributed algorithm to the optimal solution.
We use a four-node network shown in Fig.~\ref{fig:initialtopo}. Node $A$ owes \$$50$ to $B$ and $C$,
node $B$ owes \$$20$ to $C$, node $C$ owes \$$80$ to $A$, and node $D$ owes \$$10$ to $C$.
Each node has \$$1$ on hand. After all the clearing payments, the borrower-lender network reduces
to Fig.~\ref{fig:reducetopo1}. Without any external financial
support, nodes $A$, $C$, and $D$ are in default, and the total amount of unpaid liability is \$$98$.
Assume that $w_i=0.45$ for $i=1,2,\cdots,N$ in LP~(\ref{eq:LP}), i.e., that each dollar of unpaid
liability contributes $0.45$ to the cost. Without any external cash injection, the value of
the cost function is $98\times 0.45 = 44.1$.

\begin{figure}[t]
\centering
\subfigure[Initial network.] {
\label{fig:initialtopo}
\includegraphics[width=0.3\columnwidth]{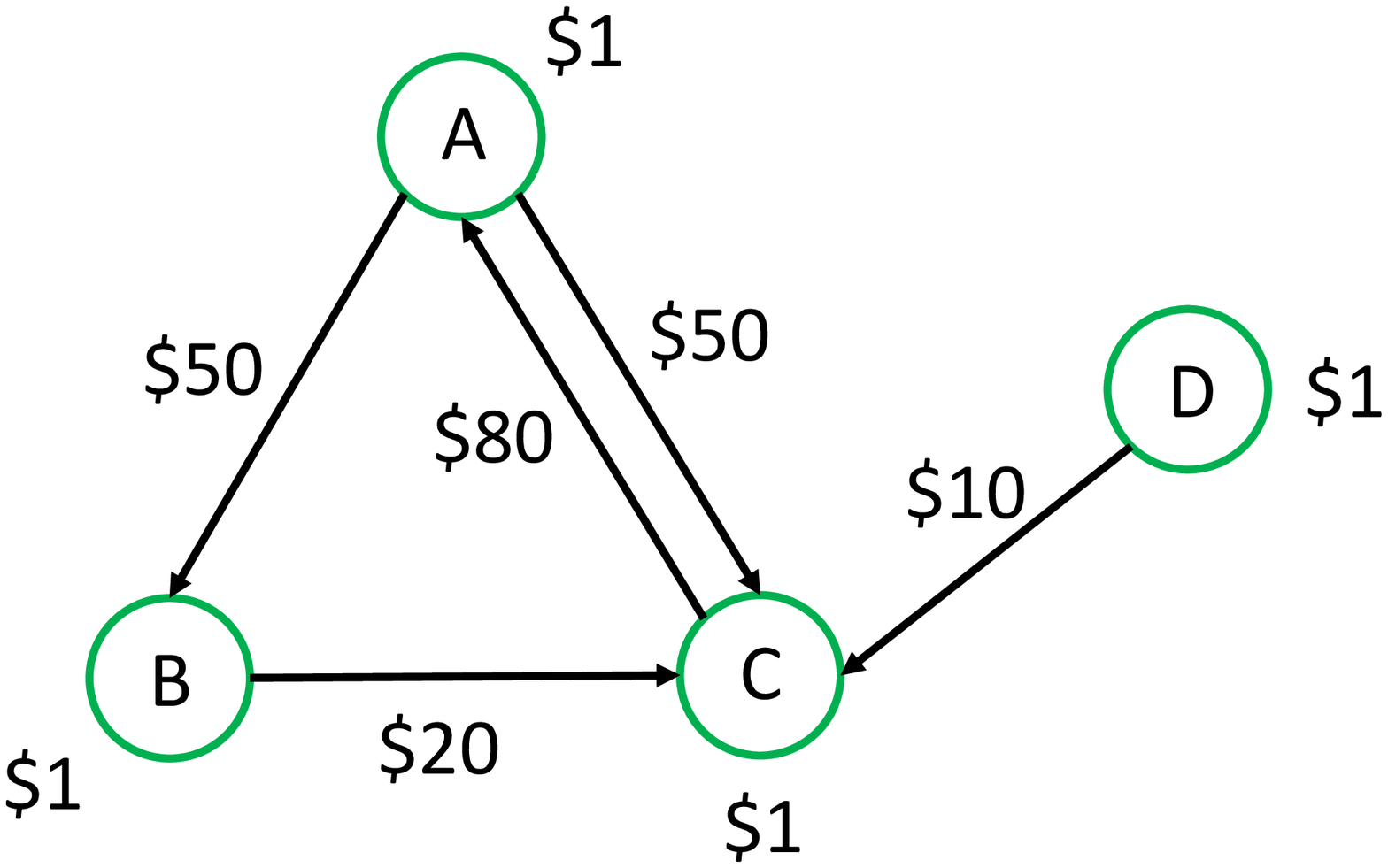}
}
\subfigure[Clearing without any bailout.] {
\label{fig:reducetopo1}
\includegraphics[width=0.3\columnwidth]{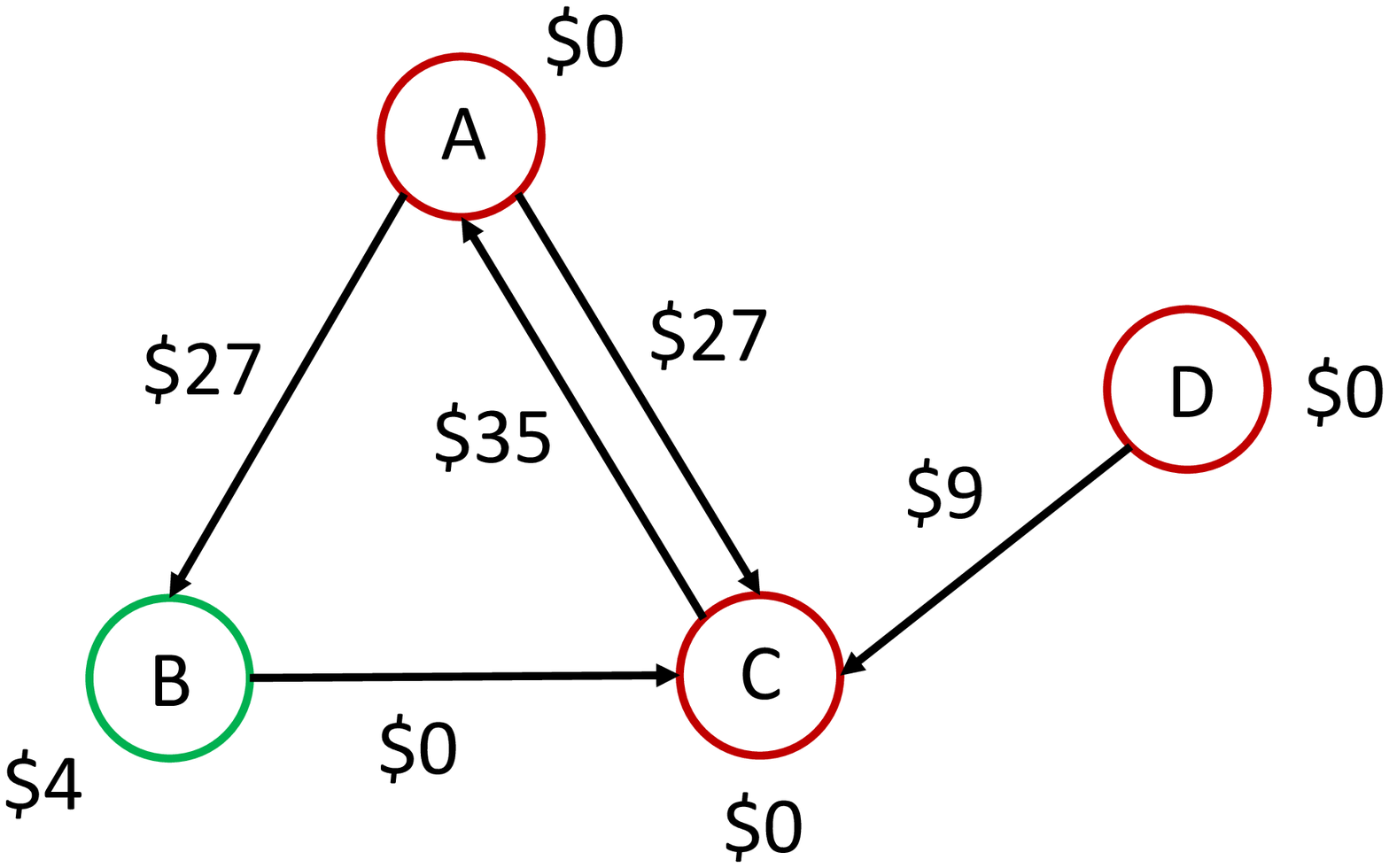}
}
\caption{A four-node network.}
\label{fig:topo}

\centering
\includegraphics[width=0.3\textwidth]{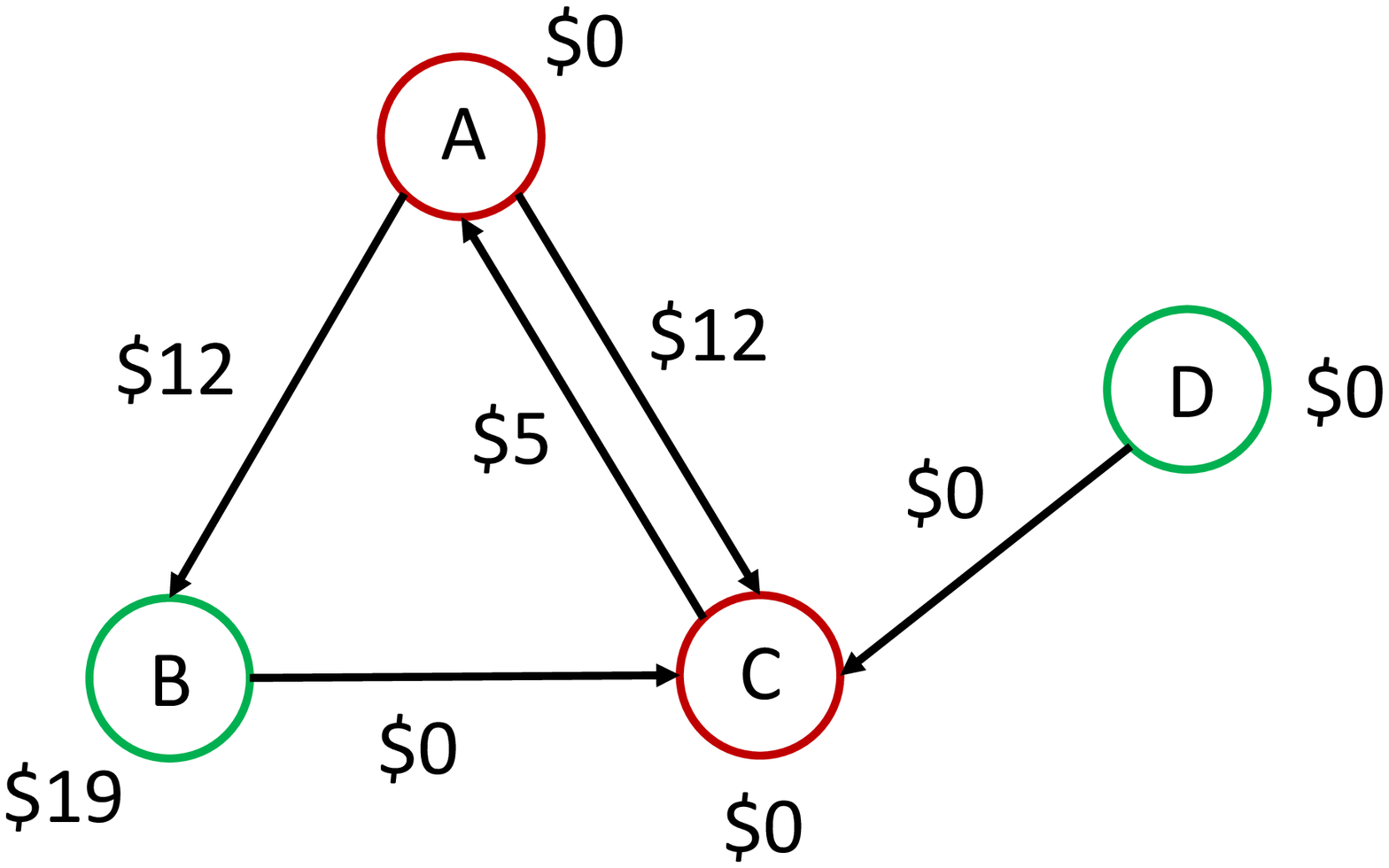}
\caption{Clearing in the network of Fig.~\ref{fig:topo} for the optimal allocation of a \$15 cash injection.}
\label{fig:reducetopobdd}

\centering
\subfigure[Payment vector.] {
\label{fig:paymentvector-bdd}
\includegraphics[width=0.4\columnwidth]{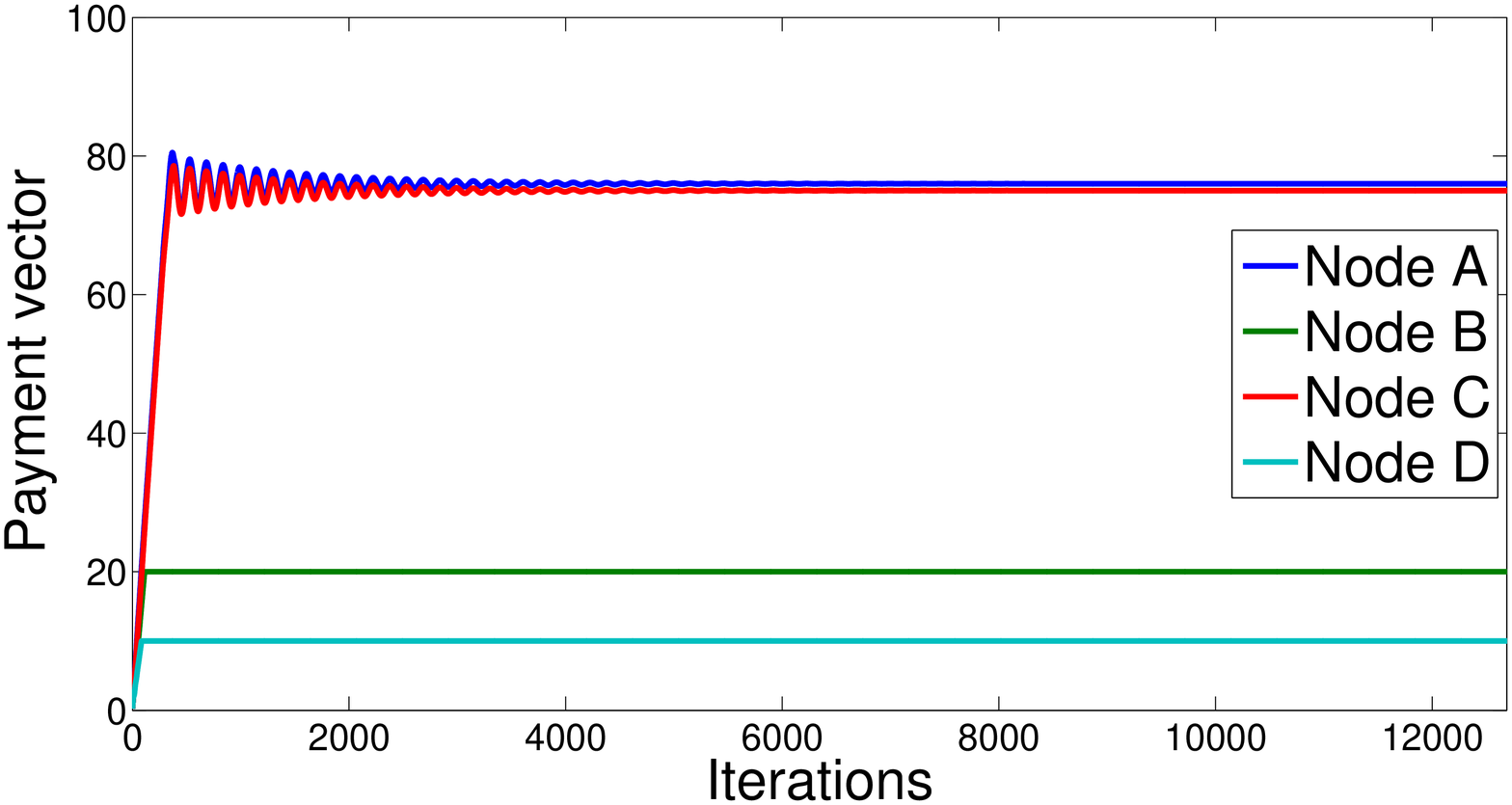}
}
\subfigure[Cash injection vector.] {
\label{fig:cashvector-bdd}
\includegraphics[width=0.4\columnwidth]{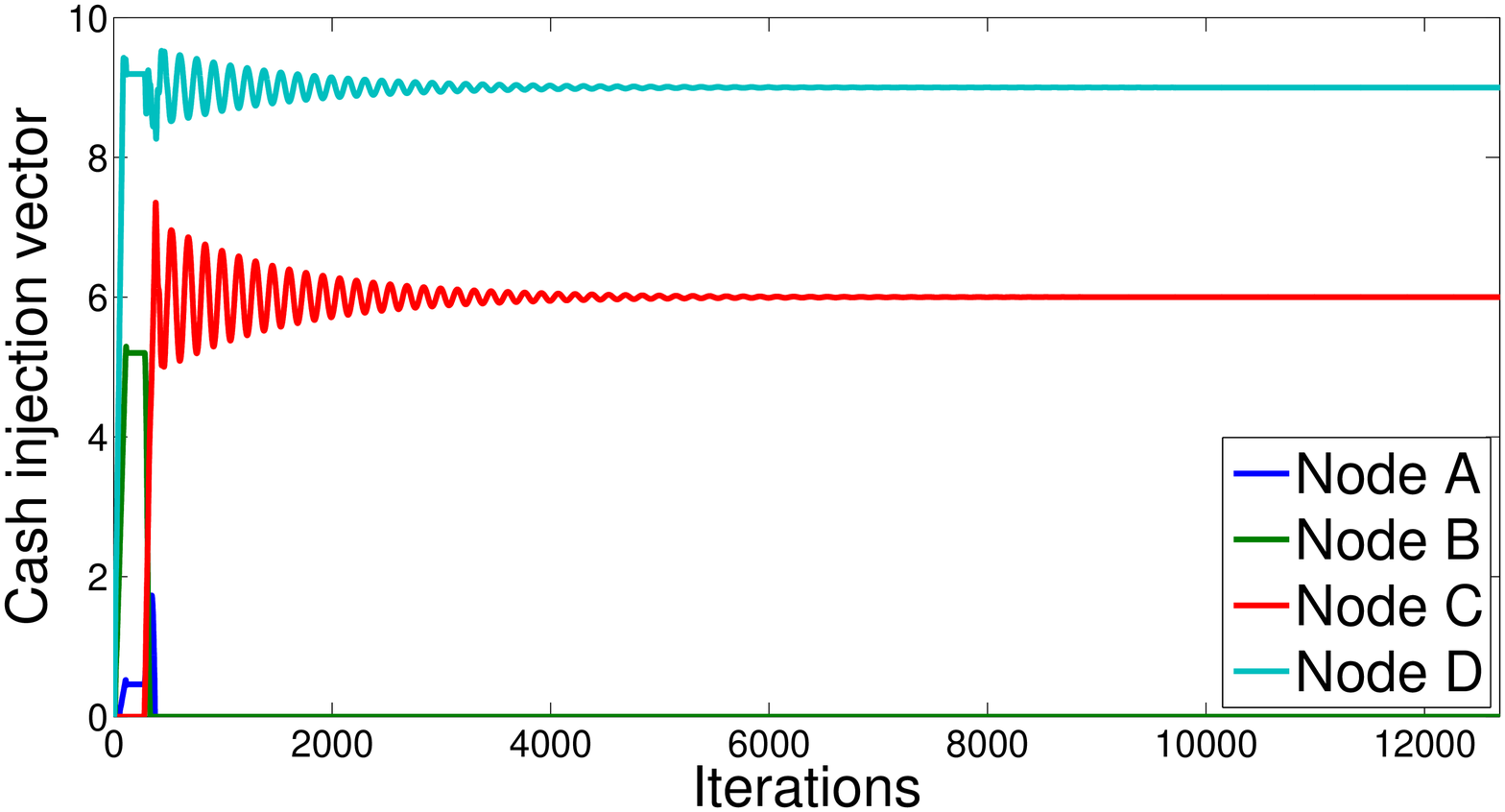}
}
\caption{ Evolution of the node payments and cash injections
through the iterations of the distributed algorithm for finding the optimal allocation of a
\$15 cash injection into the network of Fig.~\ref{fig:topo}.}
\label{fig:dis-result-bdd}
\end{figure}

We first study Problem~I, the case with a fixed maximum total amount of injected cash.
We assume that we can inject at most $\$15$ into the system. We run our algorithm
with initial ${\bf y}(0)={\bf z}(0)={\bf q}(0)=\mathbf{0}$ and $\lambda=0$.
The step size is $\alpha=\beta=0.1$, and the stopping tolerance is $\delta_1=\delta_2=10^{-6}$.
Figs.~\ref{fig:paymentvector-bdd} and~\ref{fig:cashvector-bdd}
illustrate the evolution of payment vector ${\bf p}$ and cash injection vector
${\bf c}$, respectively, as a function of the number of iterations of
the proposed distributed algorithm.
The payment vector converges to $[p_A,p_B,p_C, p_D]=[76,20,75,10]$;
the cash injection vector converges to $[c_A,c_B,c_C,c_D]=[0,0,6.0,9.0]$.
These are optimal, as verified by solving the LP~(\ref{eq:LP}-\ref{eq:LP-constraint2}) directly. With
external cash injection, the borrower-lender network reduces to Fig.~\ref{fig:reducetopobdd}
after all the payments. Now the total unpaid liability is \$$29$.
Thus the value of the cost due to unpaid liability after the optimal bailout
is $29\times 0.45=13.05$.

\begin{figure}[t]
    \centering
    \includegraphics[width=0.3\textwidth]{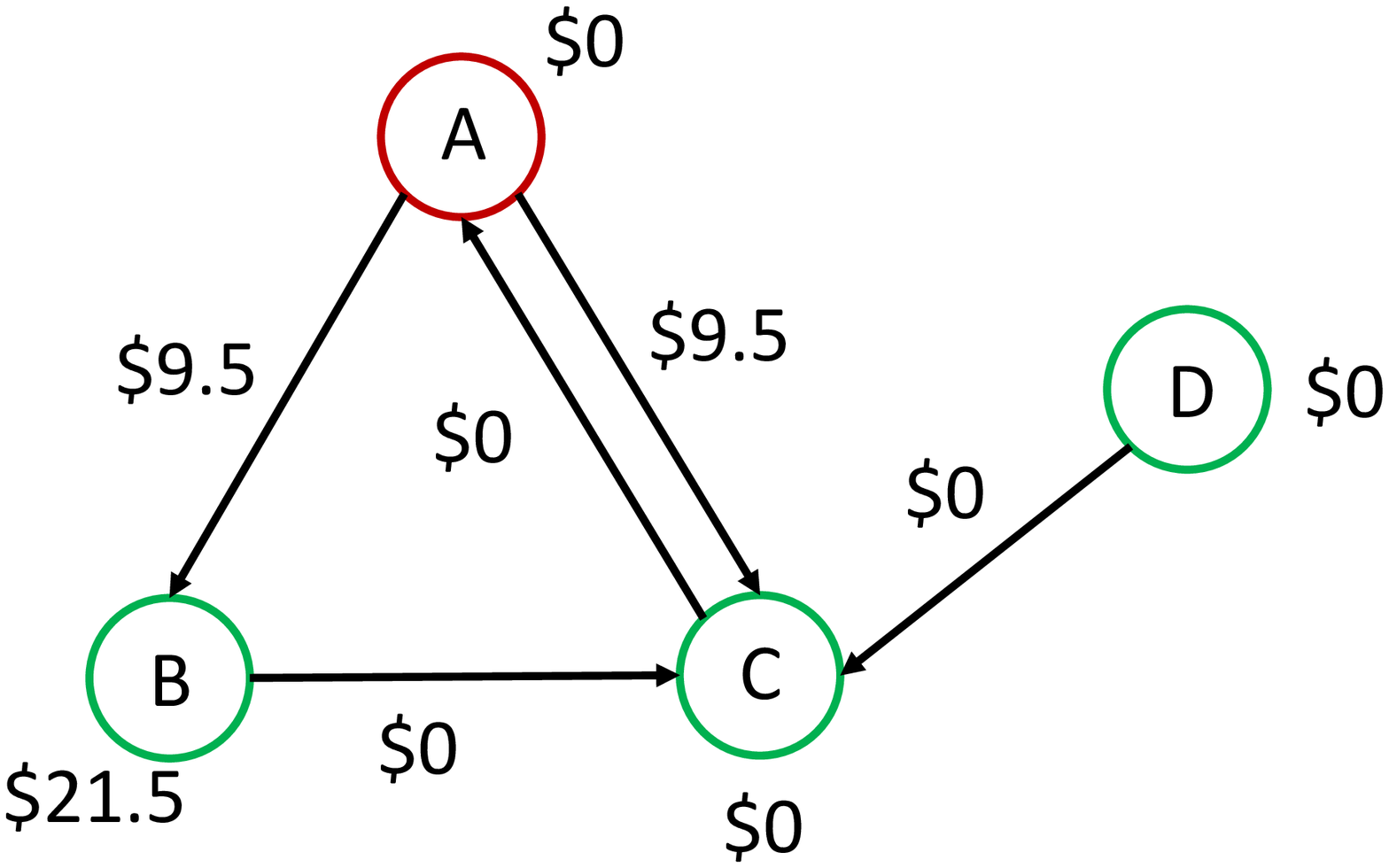}
    \caption{Clearing with the optimal bailout amount and allocation.}
\label{fig:reducetopoubdd}

\centering
\subfigure[Payment vector.] {
\label{fig:paymentvector}
\includegraphics[width=0.4\columnwidth]{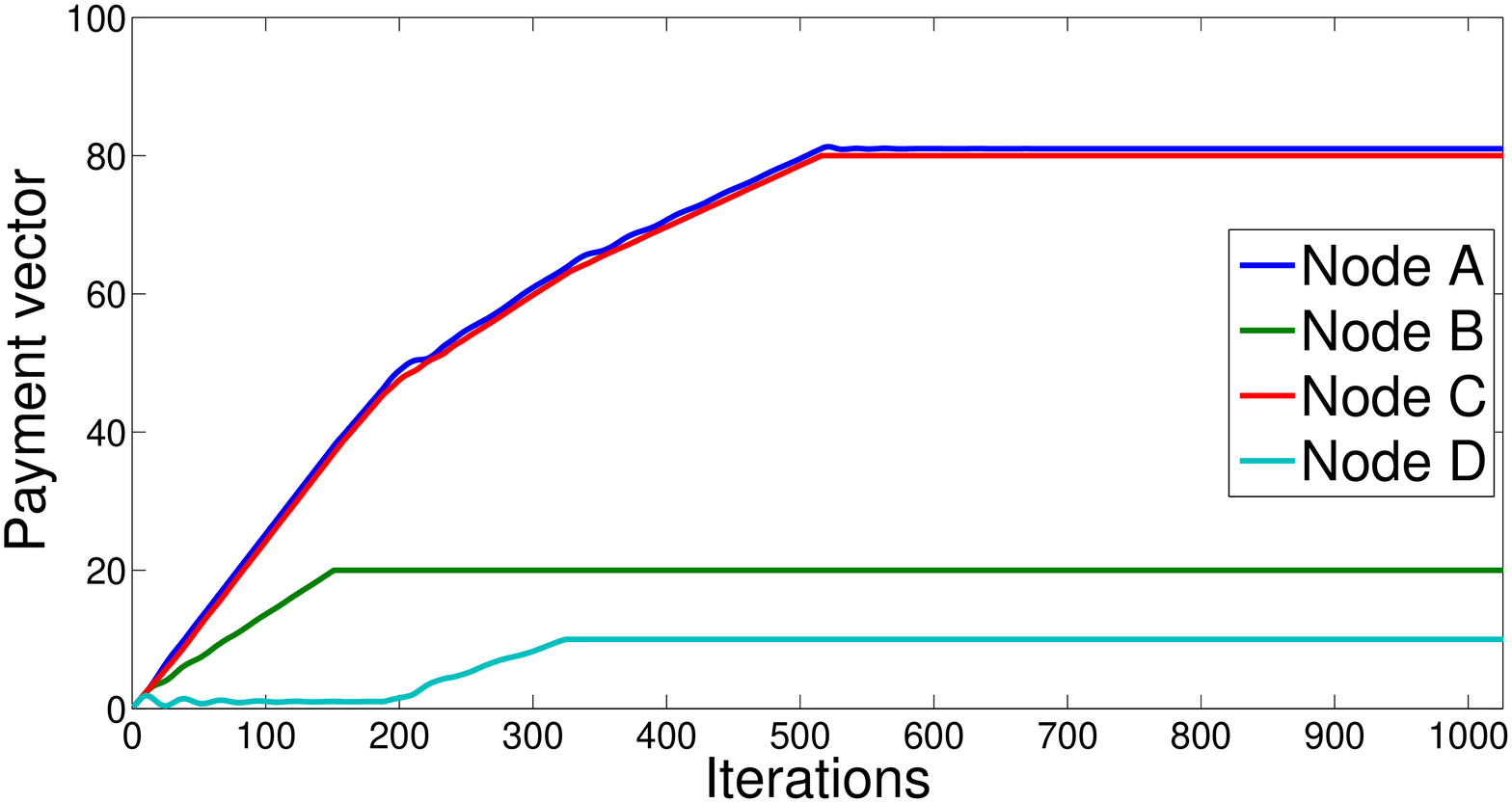}
}
\subfigure[Cash injection vector.] {
\label{fig:cashvector}
\includegraphics[width=0.4\columnwidth]{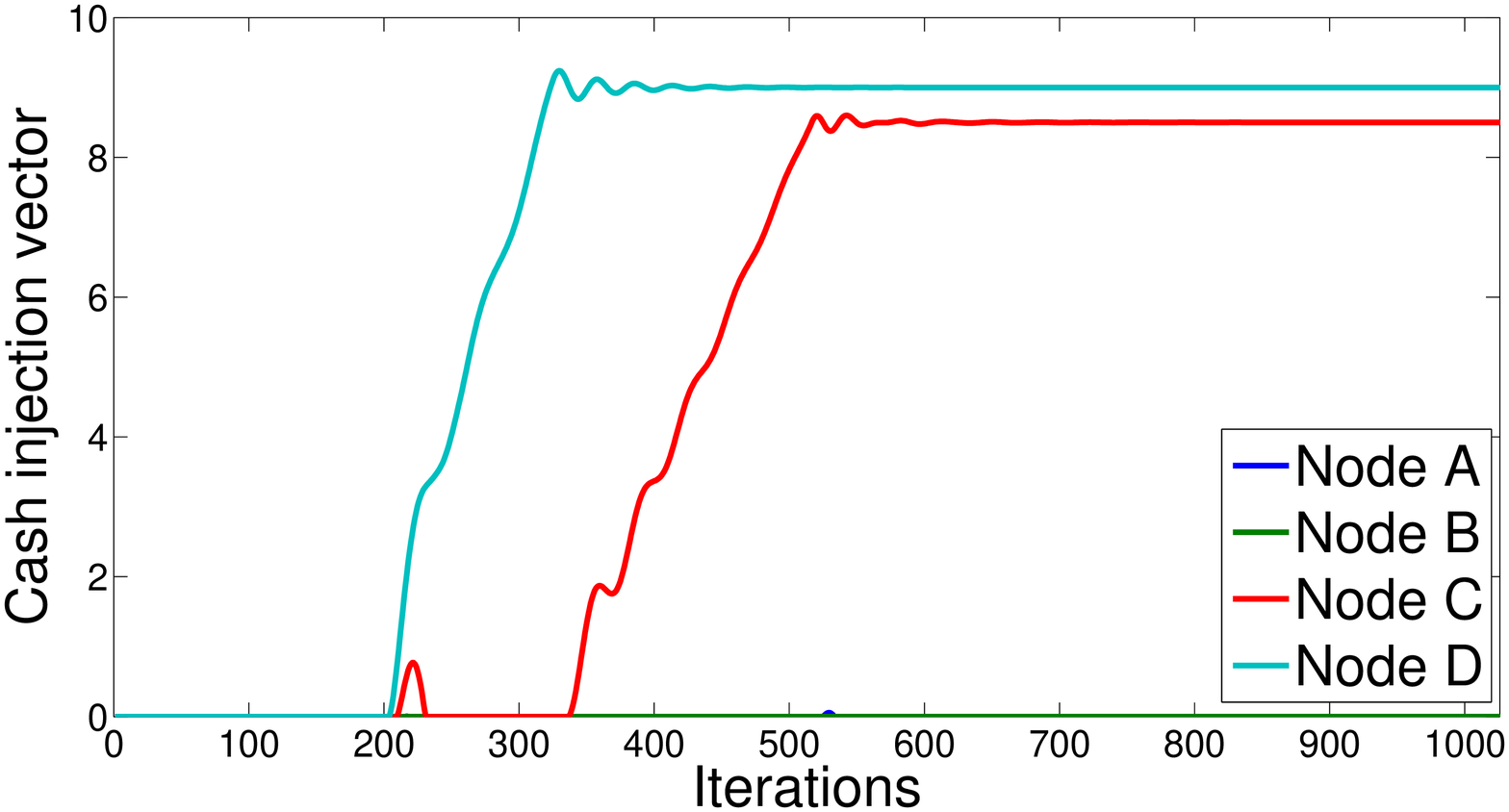}
}
\caption{ Evolution of the node payments and cash injections
through the iterations of the distributed algorithm that optimizes both the amount and the allocation
of the injected cash.}
\label{fig:dis-result}
\end{figure}

Second, we test our algorithm on the Lagrange formulation of Problem~I. In this example,
the initial settings are the same as in the previous example. In addition, we fix the Lagrange
multiplier $\lambda=1$. As shown in Fig.~\ref{fig:paymentvector} and Fig.~\ref{fig:cashvector},
the payment vector converges to $[p_A,p_B,p_C, p_D]=[81,20,80,10]$, and the cash injection vector
converges to $[c_A,c_B,c_C,c_D]=[0,0,8.5,9.0]$. These are optimal, as verified by solving the
LP~(\ref{eq:LP2}) directly. With external cash injection, the borrower-lender network reduces
to Fig.~\ref{fig:reducetopoubdd} after all the payments. Now the total unpaid liability
is \$$19$, and the cash injection amount is \$$17.5$. Thus the value of the cost function
after the optimal bailout is $17.5+19\times 0.45 =26.05$. By injecting \$17.5, we reduce
the total unpaid liability by \$35.55, and we reduce the total cost by $\$35.55-\$17.5=\$18.05$.

We see from Fig.~\ref{fig:reducetopoubdd} that although $A$ is still in
default, in the optimal bailout strategy we choose not to inject any
cash in $A$. The reason is that if we inject some cash \$$x$ into $A$ in
Fig.~\ref{fig:reducetopoubdd}, the total unpaid liability will decrease by
\$$x$ so that the unpaid liability term of the cost function will be reduced
by $0.45x$, i.e., the value of the overall cost function will actually increase by
$x-0.45x=0.55x$.

\subsubsection{Example 2: A Core-Periphery Network}
In this section, we examine the practicality of our distributed algorithm. As in
Section~\ref{sec:linear_bankruptcy_costs}, we assume that the US interbank network is well modeled as 
a core-periphery network that consists of a core of 15 highly interconnected banks 
to which most other banks connect~\cite{SoBeArGlBe07}.
We test the distributed algorithm for LP~(\ref{eq:LP2}) on a simulated core-periphery network
illustrated in Fig.~\ref{fig:cp-sim}. The core network consists of 15 fully connected core nodes.
Each core node has 70 corresponding periphery nodes which owe money only to this core node.
For each pair of two core nodes $i$ and $j$, we set $L_{ij}$ as a random number uniformly distributed
in $[0,10]$. For a core node $i$ and its periphery node $k$, $L_{ki}$ is set to be uniformly distributed
in $[0,1]$. All these obligation amounts are statistically independent.
The asset vector is ${\bf e}=\mathbf{0}$. In addition, we assume $w_i=0.3$ for $i=1,2,\cdots,N$,
and $\lambda=1$ in LP~(\ref{eq:LP2}). We generate 100 independent samples of a core-periphery network
drawn from this distribution.  These samples thus all have the same topology but different amounts
of liabilities. We run the distributed algorithm of Section~\ref{subsubsec:L4Prob1A} with
initial conditions ${\bf y}(0)={\bf z}(0)={\bf q}(0)=\mathbf{0}$. The step size is $\beta=0.01$. 

The stopping criterion for the distributed algorithm is
$\max\{\Vert\tilde{\bf y}(t+1)-\tilde{\bf y}(t)\Vert_\infty,\Vert\tilde{\bf z}(t+1)-\tilde{\bf z}(t)\Vert_\infty\} < 10^{-7}$.
Let $T_d$ be the value of the total cost function $W+\lambda C$ calculated by our distributed algorithm,
and let $T_l$ be the corresponding value obtained by solving the linear program directly, in a centralized
fashion. Under this stopping criterion, the relative error, defined as $|T_d-T_l|/T_l$, is less than
$10^{-6}$ for each sample in our simulations.


\begin{figure}
    \centering
    \includegraphics[width=0.7\textwidth]{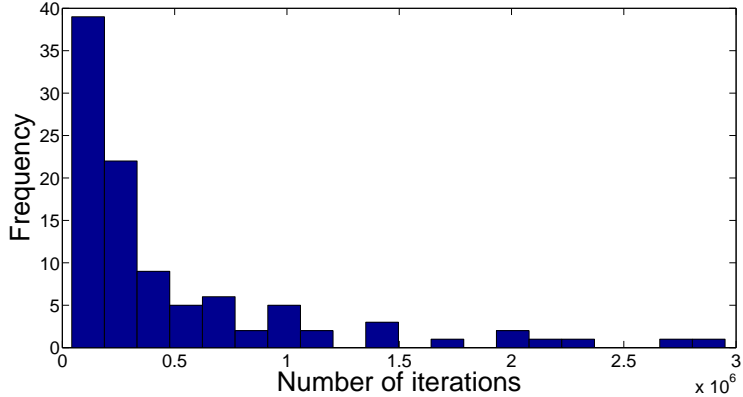}
    \caption{Number of iterations for the core-periphery network with $\delta_1=\delta_2=10^{-7}$.} \label{fig:iter}
\end{figure}

The number of iterations is shown in Fig.~\ref{fig:iter}. The average number of iterations is
$4.98\times 10^{5}$. Moreover, from Fig.~\ref{fig:iter}, we can see that for most cases,
the algorithm terminates within $10^6$ iterations.

The time spent on each iteration consists of two parts: the computing time and the time it takes to 
convey messages between the nodes.  During each iteration, a node needs to transmit information
to a set of neighbors twice: in Steps 1 and 2. Note that in Step 3, the stopping sign $b_i$
is transmitted to the central node. However, it is not necessary for a node to wait for the response
before next iteration.  Therefore, we do not count it towards the communication delay during
one iteration. It takes light 13.2ms to travel from LA to NYC, which is the longest possible
distance between two financial institutions within the continental US. So the propagation delay
in one iteration could be roughly estimated as $13.2ms\times 2=26.4ms$. Hence, for most cases,
the algorithm would terminate within $26.4ms\times 10^6=7.3h$, and the average running time 
would be below $26.4ms\times 4.98\times 10^{5}=3.65h$. These running times would be acceptable
in applications where these computations are run overnight or during a weekend.
Note that the computation time at each node is negligible compared to these communication times,
and therefore we ignore it in these estimates.

Another possible set-up is that each institution provides a client-end computer and we colocate
these computers in one room. Assuming that the longest network cable in this room is 100 meters,
the propagation delay per iteration would be around $2\times 100/(3\times 10^8)=6.67\times 10^{-7}s$.
For the computing time, we just analyze the core nodes because the periphery nodes have no borrowers
and only one creditor so that the computing time for the periphery nodes is much smaller than for
the core nodes. Usually, multiplications dominate the computing time. At each iteration, a core node
calculates $q_j(t)\Pi_{ij}$, $\Pi_{ij}p_i(t)$, and $q_j(t+1)\Pi_{ij}$ for all its creditors $j$.
Since the core network is a fully connected network with 15 core nodes, a core node has 14 creditors
so that it does less than 50 multiplications per iteration. Assuming that each multiplication takes
500 cpu cycles and the cpu on the client-end computer is 3GHz, then the computing time per iteration
is around $50\times 500 / (3\times 10^9)=8.33\times 10^{-6}s$. Thus, for most cases, the algorithm terminates
within $(8.33\times 10^{-6}+6.67\times 10^{-7})\times 10^6 \approx 10s$. By colocating the
client-end computers of all the financial institutions in the system, we can significantly reduce
the running time of our distributed algorithm so that it can be easily run many times during a day.

In a monitoring application, our aim might be to calculate the payments approximately rather than
exactly. In this case, the running time can be reduced by relaxing the termination tolerance.
We set the stopping criterion as $\max\{\Vert\tilde{\bf y}(t+1)-\tilde{\bf y}(t)\Vert_\infty,\Vert\tilde{\bf z}(t+1)-\tilde{\bf z}(t)\Vert_\infty\} < 10^{-3}$. Under this stopping criterion, the relative error,
$|T_d-T_l|/T_l$, is around $1\%$ for each sample in our simulations.
Fig.~\ref{fig:iter-loose} illustrates the number of iterations. The average number is $4260$.
The number of iterations is less than 10000 for most cases. By similar analysis, the average
running time for the non-colocated scenario is $26.4ms\times 4260\approx 2 min$. For most cases,
the algorithm will be terminated within $26.4ms\times 10^4\approx 4.4 min$. If we colocate
the client-end computers of all the financial institutions, the algorithm will terminate
within $(8.33\times 10^{-6}+6.67\times 10^{-7})\times 10^4 \approx 0.1s$.

\begin{figure}
    \centering
    \includegraphics[width=0.7\textwidth]{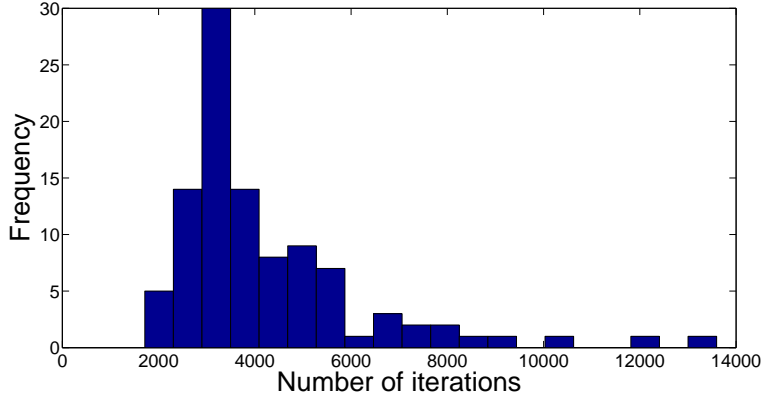}
    \caption{Number of iterations for the core-periphery network with $\delta_1=\delta_2=10^{-3}$.} \label{fig:iter-loose}
\end{figure}

The above running time analysis is for the Lagrangian formulation of Problem~I, in which $\lambda$ is
a constant. In Problem~I, $\lambda$ is a dual variable that also needs to converge. So the running time
of the distributed algorithm for Problem~I will be larger than the time for its Lagrangian formulation.
From Fig.~\ref{fig:dis-result-bdd} and Fig.~\ref{fig:dis-result}, we observe that with the same stopping
tolerance, the number of iterations of the distributed algorithm for Problem~I is around 10 times
the number of iterations for its Lagrangian formulation. Therefore, for Problem~I, to calculate
the exact payment vector, the algorithm will terminate within around $70h$ for the non-colocated scenario
and within $100s$ for the colocated scenario.  To obtain the payments within 1\% error,
the algorithm will terminate within around $44min$ and $1s$ for the non-colocated and colocated scenarios.

\section{Conclusions} \label{sec:conclusion}
In this work, we have developed a linear program to obtain the optimal cash injection policy,
minimizing the weighted sum of unpaid liabilities, in a one-period, non-dynamic financial system.
We have further proposed a reweighted $\ell_1$ minimization algorithm based on this linear program
and a greedy algorithm to find the cash injection allocation strategy which minimizes the number of defaults in the system.  By constructing three topologies in which the optimal solution can be
calculated directly, we have tested both two algorithms and shown through simulation that the results
of the reweighted $\ell_1$ minimization algorithm are close to optimal, 
and the performance of the greedy algorithm highly depends on the network topology.  
We also compare these two algorithms using three types of random networks for which the optimal
solution is not available.
Moreover, we have proposed a duality-based distributed algorithm to solve
the linear program. The distributed algorithm is iterative and is based on message passing between each node
and its neighbors. No centralized gathering of large amounts of data is required, and each
participating institution avoids revealing its proprietary book information to other institutions.
The convergence and the practicality of the distributed algorithm are both supported by our simulations.
We have also considered the situation where the capital of institutions at maturity is a random vector
with known distribution. We have developed a stochastic linear program to find the optimal cash injection
policy to minimize the expectation of the weighted sum of unpaid liabilities. To solve it, we have proposed
two algorithms based on Monte Carlo sampling: Benders decomposition algorithm and projected stochastic
gradient descent.
In addition, we show that the introduction of the all-or-nothing payment mechanism
turns the optimal cash injection allocation
problem into an NP-hard mixed-integer linear program. However, we show through simulations that use
optimization package CVX~\cite{cvx,GrBo08} that this problem can be accurately solved in a few seconds 
for a network size comparable to the size of the US banking network.

\newpage
\begin{appendices}
\section{Comparison of Three Algorithms for Computing the Clearing Payment Vector} \label{app:3algo}
\subsection{Proportional Payment Mechanism}
In~\cite{EiNo01}, zero bankruptcy costs are assumed, and
three methods of finding the clearing payment vector are proposed:
a fixed-point algorithm, the fictitious default algorithm and an optimization method.
In this section, we first introduce and analyze these three methods and then compare
their computation times under different network topologies.

\subsubsection{Fixed-Point Algorithm}
By definition, the clearing payment vector is a fixed point of the following map:
\begin{equation}
\Phi({\bf p})=\min\{(\Pi^T {\bf p}+{\bf e}), \bar{\bf p}\}. \nonumber
\end{equation}
where the minimum of the two vectors is component-wise.  Under certain mild assumptions
specified in~\cite{EiNo01}, the fixed point is unique.  It can be found
iteratively via the following algorithm~\cite{EiNo01}.\\
\emph{Fixed-point algorithm:}
\begin{enumerate}
\item
Initialization: set ${\bf p}^0 \leftarrow\bar{\bf p}$, $k\leftarrow 0$, and set the stopping
tolerance $\delta_0$ to a small positive number based on the accuracy requirement.
\item
${\bf p}^{k+1}\leftarrow\Phi({\bf p}^k)$.
\item
If $\Vert{\bf p}^{k+1}-{\bf p}^k\Vert_\infty < \delta_0$, stop and output the clearing
payment vector ${\bf p}^{k+1}$; else, set $k\leftarrow k+1$ and go to Step 2.
\end{enumerate}
At each iteration, the computational complexity is dominated by $\Pi^T {\bf p}$, which is $\Theta(N^2)$.
The number of iterations is highly dependent on the network topology and the amounts of liabilities.

\subsubsection{Fictitious Default Algorithm}
\label{sec:fictitious_default}
The fictitious default algorithm is proposed in Section 3.1 in~\cite{EiNo01}. The basic idea
is to first assume that all the nodes pay their liabilities in full.  If, under this assumption,
every node has enough funds to pay in full, then the algorithm terminates.  If some nodes do not
have enough funds to pay in full, it means that these nodes would default even if all the other
nodes pay in full.  Such defaults that are identified during the first iteration of the algorithm
are called {\em first-order defaults}. In the second iteration, we assume that only the first-order
defaults occur. Every non-defaulting node $k$ pays in full, i.e., $p_k=\bar{p}_k$; every defaulting
node $i$ pays all its available funds, i.e., $p_i=\displaystyle \sum_{j=1}^N \Pi_{ji} p_j + e_i$. 
If there is
no new defaulting nodes during this second iteration, then the algorithm is terminated.  Otherwise,
the new defaulting nodes are called {\em second-order defaults}, and we proceed to the third iteration.
In the third iteration we assume that both the first-order and second-order defaults occur.  We
calculate the new payment vector and again check the set of defaulting nodes.  We keep iterating
until no new defaults occur. Since there are $N$ nodes in the system, this algorithm is guaranteed
to terminate within $N$ iterations. The specifics of the fictitious default algorithm are as follows.\\
\emph{Fictitious default algorithm:}
\begin{enumerate}
\item Initialization: ${\bf p}^1 \leftarrow \bar{\bf p}$, $k \leftarrow 1$, and
$D^{(0)} \leftarrow \varnothing$.
\item For all nodes $i$, compute the difference between their incoming payments and their
obligations:
\[
v_i^{(k)} \leftarrow \sum_{j=1}^N \Pi_{ji} p_j^{(k)} + e_i - \bar{p}_i
\]
\item Define $D^{(k)}$ as the set of defaulting nodes:
\[
D^{(k)} = \left\{i:\,\,v_i^{(k)}<0\right\}.
\]
\item If $D^{(k)} = D^{(k-1)}$, terminate.
\item Otherwise, set $p_i^{(k+1)} \leftarrow \bar{p}_i$ for all $i\not\in D^{(k)}$.
For all $i\in D^{(k)}$, compute the payments $p_i^{(k+1)}$ by solving the following system of equations:
\[
p_i^{(k+1)} = e_i + \sum_{j\in D^{(k)}} \Pi_{ji} p_j^{(k+1)} + \sum_{j\not\in D^{(k)}} \Pi_{ji}
\bar{p}_j, \mbox{ for all } i\in D^{(k)}
\]
\item Set $k\leftarrow k+1$ and go to Step 2.
\end{enumerate}

At each iteration of the fictitious default algorithm, the computational complexity is dominated
by solving the linear equations in Step 5. The number of unknowns in these equations and the number
of equations are both equal to the number of elements in
$D^{(k)}$. In the worst case, the number of defaulting nodes in the system is of the same
order as $N$. In this case the computational complexity per iteration is $O(N^3)$~\cite{Ye91}. Compared to
the fixed-point algorithm, the fictitious default algorithm has a larger computational complexity per
iteration, and, as shown below in Section~\ref{subsubsec:clearing_vector_time}, larger running times
on several network topologies.
However, the advantage of the fictitious default algorithm is that it is guaranteed to terminate
within $N$ iterations. Moreover, the fictitious default algorithm will produce the exact value
of clearing payment, unlike the fixed-point algorithm which produces an approximation.

\subsubsection{Linear Programming Method}
Define $f({\bf p})=\displaystyle \sum_{i=1}^{N} p_i = \mathbf{1}^T{\bf p}$, which is a strictly
increasing function of $\bf p$. By Lemma 4 in~\cite{EiNo01}, the clearing payment vector can be
obtained via the following linear program:
\begin{align}
& \displaystyle \max_{\bf p} \mathbf{1}^T{\bf p} \label{eq:LPmethod} \\
& \mbox{subject to } \nonumber\\
& \mathbf{0}\leq {\bf p} \leq \bar{\bf p}, \nonumber\\
& {\bf p} \leq \Pi^T{\bf p} + {\bf e}. \nonumber
\end{align}
The computational complexity of solving an LP is $O(N^3)$~\cite{Ye91}.

\begin{table}
\caption{Comparison of the running times for the computation of the clearing payment vector under
the proportional payment mechanism using the fixed-point algorithm, fictitious default algorithm,
and linear programming.}
\begin{center}
\begin{tabular}{|c|c|c|c|c|c|c|}
\hline
 & \multicolumn{2}{c|}{Fixed-point algorithm} & \multicolumn{2}{c|}{Fictitious default algorithm} & \multicolumn{2}{c|}{LP method}\\
\hline
 & ave (s) & stdev & ave (s) & stdev & ave (s) & stdev\\
\hline
fully connected & 0.9128 & 0.1045 & 10.7341 & 0.7182 & 53.1725 & 11.8947\\
\hline
core-periphery & 0.0869 & 0.0342 & 7.8213 & 1.2843 & 0.1964 & 0.0507\\
\hline
linear chain & 0.0462 & 0.0170 & 10.2574 & 1.0211 & 0.1610 & 0.0449\\
\hline
\end{tabular}
\vspace*{-0.2in}\\
\end{center}
\label{tb:3methods}
\end{table}

\subsubsection{Comparison of Running Times on Three Different Topologies}
\label{subsubsec:clearing_vector_time}
We calculate the clearing payment vector via the above three methods on three different network
topologies and compare the running times. 
The first network topology is a fully connected network with 1000 nodes. All the obligation amounts
$L_{ij}$ and asset amounts $e_i$ are independent random variables, uniformly distributed in $[0,1]$.
The second network topology is a core-periphery network shown in Fig.~\ref{fig:cp-sim}. It contains
15 fully connected core nodes.  Each core node has 70 periphery nodes. Each periphery node has
a single link pointing to the corresponding core node.
All the obligation amounts $L_{i,j}$ are independent uniform random variables.  For each pair of core
nodes $i$ and $j$ the obligation amount $L_{ij}$ is uniformly distributed in $[0,10]$. For a core node $i$
and its periphery node $k$, the obligation amount $L_{ki}$ is uniformly distributed in $[0,1]$.
The asset amounts $e_i$ are uniformly distributed in $[0,0.25]$.
The third network topology is a long linear chain network with 1000 nodes.
For $i=1,2,\ldots,N-1$, the obligation amount $L_{i(i+1)}$ is uniformly distributed in $[0,10]$,
and for other pairs of $i$ and $j$, $L_{ij}=0$. The asset amounts $e_i$ are uniformly distributed
in $[0,1]$.

For each type of network, we generate 100 samples. 
We run the Matlab code on a personal computer
with a 2.66GHz Intel Core2 Duo Processor P8800.
The average running times and the sample standard
deviations of the running times for the three methods
are shown in Table~\ref{tb:3methods}.  For all three types of networks,
the fixed-point algorithm is the most efficient one.  Note that the computation time
of linear program method
is highly variable because simpler topologies result in $\Pi$ being a sparse matrix, reducing
the running time.

\subsection{All-or-Nothing Payment Mechanism}
\subsubsection{Fixed-Point Algorithm and Fictitious Default Algorithm}
We now assume the all-or-nothing payment mechanism
where node $i$ pays $\bar{p}_i$ if it is solvent and pays nothing if it defaults. Therefore,
the clearing payment vector is a fixed point of the map $\Psi$ defined as follows:
\[
\Psi_i({\bf p}) = \left\{
\begin{array}{cl}
\bar{p}_i & \text{if } \displaystyle \sum_{j=1}^{N} \Pi_{ji} p_j + e_i \ge \bar{p}_i\\
0 & \text{otherwise}
\end{array} \right.
\]
We find the fixed point of $\Psi(\cdot)$ iteratively via the following algorithm.\\
\emph{Fixed-point algorithm:}
\begin{enumerate}
\item
Initialization: set ${\bf p}^0 \leftarrow\bar{\bf p}$, $k\leftarrow 0$.
\item
${\bf p}^{k+1}\leftarrow\Psi({\bf p}^k)$.
\item
If ${\bf p}^{k+1}={\bf p}^k$, stop and output the clearing
payment vector ${\bf p}^{k+1}$; else, set $k\leftarrow k+1$ and go to Step 2.
\end{enumerate}

In fact, under the all-or-nothing payment mechanism, this fixed-point algorithm can be interpreted
as the following fictitious default algorithm.
We initially assume that all the nodes pay their liabilities in full, i.e., ${\bf p}^0=\bar{\bf p}$. 
If, under this assumption, every node has enough funds to pay in full, then the algorithm terminates. 
If some nodes do not have enough funds to pay in full, it means that these nodes would default even if
all the other nodes pay in full. We define these nodes
as \emph{first-order defaults}. With function $\Psi(\cdot)$, we identify the first-order defaults
and set their payments to zero. In the second iteration, we assume that only the first-order defaults occur.
Every non-defaulting node $k$ pays in full, i.e., $p_k=\bar{p}_k$; every defaulting node $i$ pays 0,
i.e., $p_i=0$. Again, with function $\Psi(\cdot)$, we identify the new defaulting nodes, which are called
\emph{second-order defaults}, and set their payments to zero. If there are no such new defaulting nodes,
the algorithm
terminates; otherwise, we proceed to the third iteration. We keep iterating until no new defaults occur,
i.e., ${\bf p}^{k+1}={\bf p}^k$.

Since there are $N$ nodes in the system, this algorithm is guaranteed to terminate within $N$ iterations.
At each iteration, the computational complexity is dominated by $\Pi^T {\bf p}$, which is $\Theta(N^2)$.
Therefore, the computational complexity of the fixed-point algorithm (fictitious default algorithm)
is $O(N^3)$.

\subsubsection{Mixed-Integer Linear Programming Method}
The clearing payment vector can also be obtained by solving MILP~(\ref{eq:MILP2}) with the assumption
that no external cash would be injected, i.e., $C=0$. With $C=0$,
MILP~(\ref{eq:MILP2}) is simplified to the following MILP:
\begin{align}
& \displaystyle \max_{{\bf p}, {\bf c}, {\bf d}} {\bf w}^T{\bf p}\label{eq:MILP3}\\
& \mbox{subject to } \nonumber\\
& p_i = \bar{p}_i(1-d_i) \mbox{, for } i=1,2,\cdots,N, \nonumber\\
& \bar{p}_i-\sum_{j=1}^{N}\Pi_{ji}p_j-e_i \le \bar{p}_i d_i \mbox{, for } i=1,2,\cdots,N, \nonumber\\
& d_i \in \{0,1\} \mbox{, for } i=1,2,\cdots,N. \nonumber
\end{align}
We solve MILP~(\ref{eq:MILP3}) via CVX~\cite{cvx,GrBo08}.

\subsubsection{Comparison of Running Times on Three Different Topologies}
We calculate the clearing payment vector under the all-or-nothing payment mechanism via the above two methods
on three network topologies described in Section~{\ref{subsubsec:clearing_vector_time}}, and compare
the running times.

Similar to Section~\ref{subsubsec:clearing_vector_time}, we generate 100 samples and
run the Matlab code on a personal computer
with a 2.66GHz Intel Core2 Duo Processor P8800.
The average running times and the sample standard deviations of the running times
for the two methods are shown in Table~\ref{tb:2methods}. For all the three topologies,
the fixed-point algorithm is significantly more efficient than the MILP method.

\begin{table}
\caption{Comparison of the running times for the computation of the clearing payment vector under
the all-or-nothing payment mechanism using the fixed-point algorithm
and mixed-integer linear programming.}
\begin{center}
\begin{tabular}{|c|c|c|c|c|c|c|}
\hline
 & \multicolumn{2}{c|}{Fixed-point algorithm / Fictitious default algorithm} & \multicolumn{2}{c|}{MILP}\\
\hline
 & ave (s) & stdev & ave (s) & stdev\\
\hline
fully connected & 0.0092 & 0.0129 & 1.2204 & 0.0909\\
\hline
core-periphery & 0.0242 & 0.0175 & 0.5338 & 0.0255\\
\hline
linear chain & 0.0279 & 0.0149 & 0.4700 & 0.0276\\
\hline
\end{tabular}
\vspace*{-0.2in}\\
\end{center}
\label{tb:2methods}
\end{table}

\section{CVX code for MILP (\ref{eq:MILP2})} \label{app:cvx}

Below is the CVX code to solve MILP (\ref{eq:MILP2}). The parameters that appear in the code
are defined in Table~\ref{tb:cvxpara}.
\begin{lstlisting}[language=matlab, frame=single]
cvx_solver mosek
cvx_begin
    variable d(N) binary % default indicator
    variable c(N) % cash injection vector
    minimize( pbar' * diag(w) * d )
    subject to
        ones(1,N) * c <= C_total;
        c >= 0;
        (Pi'-I) * diag(pbar) * d <= Pi' * pbar + e + c - pbar;
cvx_end
\end{lstlisting}

\begin{table}
\caption{Parameters in CVX codes.}
\begin{center}
\begin{tabular}{|c|c|}
\hline
Parameters in CVX codes & Notation in this paper\\
\hline
pbar & $\bar{\bf p}$\\
\hline
Pi & $\Pi$\\
\hline
e & {\bf e}\\
\hline
c & {\bf c}\\
\hline
w & {\bf w}\\
\hline
d & {\bf d}\\
\hline
I & $I_{N\times N}$\\
\hline
C\_total & $C$\\
\hline
\end{tabular}
\vspace*{-0.2in}\\
\end{center}
\label{tb:cvxpara}
\end{table}

\section{A fast algorithm for quadratic program (\ref{eq:quadprog})} \label{app:quadprog}
\begin{align}
& \displaystyle \min_{\bf c} \Vert{\bf c}-\tilde{\bf c}^m\Vert_2^2 \label{eq:quadprog2} \\
& \mbox{subject to } \nonumber\\
& \mathbf{1}^T{\bf c} = C, \label{eq:quadprog2-constraint1}\\
& {\bf c} \geq \mathbf{0}.\label{eq:quadprog2-constraint2}
\end{align}

Let a scalar $\kappa$ and an $N \times 1$ vector $\boldsymbol \iota$ be Lagrange multipliers 
for constraint (\ref{eq:quadprog2-constraint1}) and (\ref{eq:quadprog2-constraint2}), respectively.
We define the Lagrangian as follows:
\begin{equation}
F({\bf c},\kappa,{\boldsymbol \iota}) = 
\Vert{\bf c}-\tilde{\bf c}^m\Vert_2^2 + \kappa (\mathbf{1}^T{\bf c} - C) - {\boldsymbol \iota}^T{\bf c}
\label{eq:Lagrangian-q}
\end{equation}

Then the primal and dual optimal ${\bf c}^\ast$ and $(\kappa,{\boldsymbol \iota})$ will satisfy
the following \emph{Karush-Kuhn-Tucker} (KKT) conditions:
\begin{align}
& \nabla_{\bf c} F({\bf c},\kappa,{\boldsymbol \iota}) = \mathbf{0}, \label{eq:KKT1} \\
& \mathbf{1}^T{\bf c} = C, \label{eq:KKT2} \\
& {\bf c} \geq \mathbf{0}, \label{eq:KKT3} \\
& {\boldsymbol \iota} \geq \mathbf{0}, \label{eq:KKT4} \\
& \iota_i c_i = 0 \mbox{, for } i=1,2,\cdots,N. \label{eq:KKT5}
\end{align}

From KKT constraint (\ref{eq:KKT1}), we have $c_i = \tilde{c}^m_i - \kappa/2 + \iota_i/2$, for 
$i = 1,2,\cdots,N$.
\begin{itemize}
\item If $\tilde{c}^m_i - \kappa/2 > 0$, then $c_i > 0$ since $\iota_i \geq 0$.
By KKT constraint (\ref{eq:KKT5}), $\iota_i = 0$. Thus, $c_i = \tilde{c}^m_i - \kappa/2$.
\item If $\tilde{c}^m_i - \kappa/2 < 0$, $\iota_i$ should be positive to guarantee that $c_i \geq 0$.
Since $\iota_i c_i = 0$ and $\iota_i > 0$, $c_i = 0$.
\item If $\tilde{c}^m_i - \kappa/2 = 0$,
then $c_i = \iota_i/2$. Since $\iota_i c_i = 0$, we have $c_i = \iota_i = 0$.
\end{itemize}

To sum up,
$c_i = \max\{\tilde{c}^m_i - \kappa/2, 0\}$, for $i = 1,2,\cdots,N$, where $\kappa$ is determined by
KKT constraint (\ref{eq:KKT2}).

\end{appendices}

\newpage

\bibliographystyle{plain}
{
\bibliography{ref}
}

\end{document}